\documentclass{article}

\PassOptionsToPackage{numbers,compress,sort}{natbib}
\usepackage[preprint]{neurips_2021}

\usepackage[utf8]{inputenc} 
\usepackage[T1]{fontenc}    
\usepackage{hyperref}       
\usepackage{url}            
\usepackage{booktabs}       
\usepackage{amsfonts}       
\usepackage{amsmath}        
\usepackage{nicefrac}       
\usepackage{microtype}      
\usepackage{xcolor}         
\usepackage[shortlabels]{enumitem}
\usepackage[ruled]{algorithm2e} 
\usepackage{amsthm}
\usepackage{graphicx}
\setlist{noitemsep,topsep=5pt,parsep=3pt,partopsep=5pt,leftmargin=2em}

\renewcommand{\paragraph}[1]{\smallskip\noindent\textbf{#1}}

\newcommand{\EE}{\mathbf{E}}
\newcommand{\PP}{\mathbf{P}}
\newcommand{\ones}{\mathbf{1}}
\newcommand{\RR}{\mathbb{R}}

\newcommand{\dom}{{\rm dom}\,}
\newcommand{\ind}{\mathbb{I}} 

\newcommand{\almeve}{\textnormal{a.e.}}
\renewcommand{\citet}{\cite}

\DeclareMathOperator*{\argmax}{arg\,max}
\DeclareMathOperator*{\argmin}{arg\,min}

\newtheorem{theorem}{Theorem}
\newtheorem{defn}{Definition}

\newtheorem{lemma}{Lemma}

\title{Online Market Equilibrium with Application to Fair Division}

\author{
  Yuan Gao \\
  Columbia University \\
  \texttt{gao.yuan@columbia.edu} \\
  \And
  Christian Kroer \\
  Columbia University \\
  \texttt{christian.kroer@columbia.edu} \\
  \texttt{email} \\
  \And
  Alex Peysakhovich \\
  Facebook AI Research \\
  \texttt{alex.peys@gmail.com} \\
}

\begin{document}

\maketitle

\begin{abstract}
    Computing market equilibria is a problem of both theoretical and applied interest. Much research to date focuses on the case of static Fisher markets with full information on buyers' utility functions and item supplies. Motivated by real-world markets, we consider an online setting: individuals have linear, additive utility functions; items arrive sequentially and must be allocated and priced irrevocably. We define the notion of an online market equilibrium in such a market as time-indexed allocations and prices which guarantee buyer optimality and market clearance in hindsight. We propose a simple, scalable and interpretable allocation and pricing dynamics termed as PACE. When items are drawn i.i.d. from an unknown distribution (with a possibly continuous support), we show that PACE leads to an online market equilibrium asymptotically. In particular, PACE ensures that buyers' time-averaged utilities converge to the equilibrium utilities w.r.t. a static market with item supplies being the unknown distribution and that buyers' time-averaged expenditures converge to their per-period budget. Hence, many desirable properties of market equilibrium-based fair division such as no envy, Pareto optimality, and the proportional-share guarantee are also attained asymptotically in the online setting. Next, we extend the dynamics to handle quasilinear buyer utilities, which gives the first online algorithm for computing first-price pacing equilibria. Finally, numerical experiments on real and synthetic datasets show that the dynamics converges quickly under various metrics.
\end{abstract}

\section{Introduction} \label{sec:intro}
A market is said to be in equilibrium when supply is equal to demand. Computing prices and allocations which constitute a market equilibrium (ME) has long been a topic of interest~\citep{scarf1967computation,kantorovich1975mathematics,othman2010finding,daskalakis2009complexity,cole2017convex,kroer2019computing}. Most existing work focuses on the case of static markets. However, in this paper we consider the case of online markets where items arrive sequentially. We consider the extension of market equilibrium to this setting and provide market dynamics which quickly converge to an equilibrium in the case of online Fisher markets.

In static Fisher markets there is a fixed supply of each item, individual preferences are linear, additive, and items are divisible (or equivalently, randomization is allowed so individuals can purchase not just items but lotteries over items). In general, finding market equilibria is a hard problem~\citep{chen2009spending,vazirani2011market,othman2016complexity}. However, in static linear Fisher markets, equilibrium prices and allocations can be computed via solving the Eisenberg-Gale (EG) convex program~\citep{eisenberg1959consensus,nisan2007algorithmic}. 

We consider an online extension of Fisher markets where buyers are constantly present but items arrive one-at-a-time. Buyers' budgets are per-period and represent their respective `bidding powers' instead of being binding constraints. We extend the definition of market equilibrium to the online setting: online equilibrium allocations and prices are time-indexed and, when averaged across time, form an equilibrium in a corresponding static Fisher market where item supplies are proportional to item arrival probabilities. Due to the stochastic nature of online Fisher markets, any online algorithm can only attain an online market equilibrium \emph{asymptotically}, that is, the allocations and prices approximately satisfy the equilibrium conditions after running the algorithm for a long time.

We propose market dynamics that find these equilibria in an online fashion based on the \emph{dual averaging} algorithm applied to a reformulation of the dual of the EG convex program. We refer to this mechanism as \textbf{PACE} (Pace According to Current Estimated utility). 
In PACE, each buyer is assigned a utility \emph{pacing} multiplier at time $0$. 
When an item arrives, the individual with the highest adjusted utility (its valuation times the multiplier) receives that item and pays a price equal to its adjusted utility. 
The pacing multipliers of all individuals are then adjusted according to a closed-form rule which is given by the time average of the subgradient of the dual of the EG program. 
Intuitively, the pacing multipliers of those that did not receive the item go up while the receiver's typically (but not always) goes down. We show that PACE yields item allocations and prices that satisfy various equilibrium properties asymptotically, for example no-regret and envy-freeness.

One important application of market equilibrium is fair allocation using the \emph{competitive equilibrium from equal incomes} (CEEI) mechanism~\citep{varian1974equity,budish2011combinatorial}. In CEEI, each individual is given an endowment of faux currency and reports her valuations for items; then, a market equilibrium is computed and the items are allocated accordingly.
However, many fair division problems are online rather than static. These include the allocation of impressions to content in certain recommender systems~\citep{robust_blog}, workers to shifts, donations to food banks~\citep{aleksandrov2015online}, scarce compute time to requestors~\cite{ghodsi2011dominant,parkes2015beyond,kash2014no}, or blood donations to blood banks~\citep{mcelfresh2020matching}. Similarly, online advertising can also be thought of as the allocation of impressions to advertisers via a market though with a budget of real money rather than faux currency. In the static CEEI case with linear additive preferences, the resulting equilibrium outcomes (i.e. results of the EG program) have been described as ``perfect justice''~\citep{arnsperger1994envy}. 
In the online case, PACE achieves the same fair allocations as CEEI asymptotically. See Appendix \ref{app:related-work} for more related work in the areas of (static and online) equilibrium computation and fair division.

We evaluate PACE experimentally in several market datasets. Convergence to good outcomes happens quickly in experiments. Taken together our results, we conclude that PACE is an attractive algorithm for both computing online market equilibria and online fair division.

\paragraph{Main contributions.} 
We consider the problem of allocating and pricing sequentially arriving items to $n$ buyers. 
This setting is termed as an \textit{online Fisher market}.
Given a sequence of item arrivals, we define an online market equilibrium as the items' allocations and prices that, in hindsight, ensure buyer optimality and market clearance.
We propose the PACE dynamics, which can be viewed as a nontrivial instantiation of the dual averaging algorithm on a reformulation of the dual of the Eisenberg-Gale convex program.
Leveraging the convergence theory of dual averaging, we show that, when item arrivals are drawn from an (unknown) underlying distribution $s$, possibly over an infinite/continuous item space, PACE ensures the following. 
\begin{itemize}
    \item The pacing multipliers generated by PACE converge to the static equilibrium \emph{utility prices}. 
    Here, ``static'' means w.r.t. to an underlying static Fisher market.
    \item Buyers' time-averaged utilities converge to the static equilibrium utilities.
    \item Buyers' time-averaged expenditures converge to their respective budgets.
\end{itemize}
These convergences are all in mean square with rates $O((\log t)/t)$, $O((\log t)/t)$ and $O((\log t)^2/t)$, respectively, where the constants in these rates involve moderate polynomials of $n$. 
In this way, PACE generates allocations and prices that constitute an online market equilibrium in the limit. 
In particular, the allocations and prices ensure that the allocation is Pareto optimal, and buyers have no regret, no envy, and get at least their proportional share asymptotically. 
We also extend PACE to the case of quasilinear buyer utilities, which yields the first online algorithm for computing first-price pacing equilibria.
Finally, numerical experiments suggest that PACE converges much faster than its theoretical rates in terms of pacing multipliers, utilities and expenditures.



\section{Static and Online Fisher Markets} \label{sec:fisher-markets}
\paragraph{Static Fisher markets and equilibria.}  
We first introduce static Fisher markets and their equilibria. Following the recent work \cite[\S 2]{gao2020infinite}, we consider a measurable (possibly continuous) item space. 
Below are the technical preliminaries for the subsequent online setting. 
They can be skimmed through and referred back to as needed.

From now on, we define $[k] := \{1, \dots, k\}$ for any $k \in \mathbb{N} := \{0, 1,2,\dots\}$ and $\mathbb{R}_+$ ($\mathbb{R}_{++}$, resp.) as the set of nonnegative (positive, resp.) real numbers.
Let $\ind\{A\} \in \{0,1\}$ denote the indicator function of an event $A$.
\begin{enumerate}[(a)]
    \item There are $n$ buyers (individuals), each having a budget $B_i>0$. 
    \item The item space is a \emph{finite} measurable space $(\Theta, \mathcal{M}, \mu)$ with $0<\mu(\Theta) < \infty$, where $\Theta$ is a ($\mu$-)measurable subset of $\RR^d$, $\mathcal{M}$ is a $\sigma$-algebra and $\mu: \mathcal{M}\rightarrow \RR_+$ is a (finite) measure.
    From now on, $L^p$ (and $L^p_+$, resp.) denote the set of (nonnegative, resp.) $L^p$ functions on $\Theta$ for any $p\in [1, \infty]$ (including $p=\infty$).
    Below are some concrete special cases for illustration. \label{item:measurable-item-space}
    \begin{enumerate}[(i)]
        \item Finite: $\Theta = [m]$, $\mathcal{M} = 2^{[m]} = \{ A: A\subseteq [m] \}$ and $\mu(A) = \sum_{a\in A} \mu(a)$ (all $2^m$ subsets are measurable and the measure is given by a point mass on each item).
        \item Lebesgue-measurable: $\mu$ is the Lebesgue measure on $\RR^d$, $\mathcal{M}$ is the Lebesgue $\sigma$-algebra and $\Theta$ is a (Lebesgue-)measurable subset of $\RR^d$ with positive finite measure. For example, $\Theta$ can be a compact subset of $\RR^d$ with a nonempty interior.
        \item Countably infinite: $\Theta = 
        \mathbb{N}$ and $\mu(A) = \sum_{a\in A} \mu(a)$ for any $A\subseteq \mathbb{N}$, where $\mu(\mathbb{N}) < 0$. For example, $\mu(a)$ can be the probability mass of a Poisson distribution, in which case $(\mathbb{N}, \mathcal{M}, \mu)$ is a probability space.
    \end{enumerate}
    \item The \emph{supplies} of items is  $ s \in L^\infty_+$, i.e., item $\theta\in \Theta$ has supply $s(\theta)$. 
    Since $\Theta$ is compact, it is measurable with a finite measure. For the finite case $\Theta = [m]$, we have $s = (s_1, \dots, s_m) \in \RR^m_+$. 
    \item The \emph{valuation} of each buyer $i$ on all items is $v_i \in L^1_+$, i.e., buyer $i$ has valuation $v_i(\theta)$ on item $\theta\in \Theta$. For the finite case $\Theta = [m]$, we have $v_i = (v_{i1}, \dots, v_{im}) \in \RR^m_+$.
    \item For buyer $i$, an \emph{allocation} of items $x_i \in L^\infty_+$ gives a utility of 
    \[ u_i(x_i) := \langle v_i, x_i \rangle := \int_\Theta v_i(\theta) x_i(\theta) d\theta,\] 
    where the angle brackets are based on the notation of applying a bounded linear functional $x_i$ to a vector $v_i$ in the Banach space $L^1$ and the integral is the usual Lebesgue integral. For the finite case $\Theta = [m]$, we have $x_i = (x_{i1}, \dots, x_{im}) \in \RR^m_+$ and the utility is 
    \[  u_i(x_i) = \langle v_i, x_i \rangle = \sum_j v_{ij} x_{ij},\] 
    the usual Euclidean vector inner product. We will use $x\in (L^\infty_+)^n$ to denote the aggregate allocation of items to all buyers, i.e., the concatenation of all buyers' allocations. 
    \item The \emph{prices} of items are modeled as $p\in L^1_+$; in other words, the price of item $\theta\in \Theta$ is $p(\theta)$. For the finite case $\Theta = [m]$, we have $p = (p_1, \dots, p_m)\in \RR^m_+$.
    \item For a measurable item subset $A\subseteq \Theta$, let $v_i(A) := \int_A v_i(\theta) d\theta$ (and similarly for $p(A)$ and $s(A)$), the $v_i$-induced measure of $A$. For the finite case $\Theta=[m]$, for any item subset $A\subset [m]$, $v_i(A) = \sum_{j\in A} v_{ij}$ (and similarly for $p(A)$ and $s(A)$). 
    \item Without loss of generality, we assume a unit total budget $\|B\|_1 = 1$, a unit total supply $s(\Theta) = 1$ and normalized buyer valuations $\langle v_i, s\rangle = 1$. In other words, all items have a total value of $1$ for every buyer.
    \label{item:static-fm-normalization}
\end{enumerate}

\begin{defn}
Given item prices $p \in L^1_+$, the \textbf{demand} of buyer $i$ is its set of utility-maximizing allocations given the prices and budget:
\[ D_i (p) := \argmax \{ \langle v_i, x_i \rangle : x_i \in L^\infty_+,\, \langle p, x_i\rangle \leq B_i \}.\] 
The associated \textbf{utility level} $\hat{U}_i (p)$ is defined as the value of $\langle v_i, x_i\rangle$ for any $x_i \in D_i(p)$.
\label{defn:demand-set}
\end{defn}
\begin{defn}
A \textbf{market equilibrium (ME)} is an allocation-price pair $(x^*, p^*) \in (L^\infty_+)^n \times L^1_+$ such that the following holds.
\begin{enumerate}[(i)]
    \item Supply feasibility: $\sum_i x^*_i \leq s$. 
    \item Buyer optimality: $x^*_i \in D_i (p^*)$ for all $i$.
    \item Market clearance: $\langle p^*, s - \sum_i x^*_i \rangle = 0$ (any item with a positive price is fully allocated).
\end{enumerate}
\label{defn:me-static}
\end{defn}
In the above definition and subsequently, all equations involving measurable functions are understood as ``holding almost everywhere.'' For example, $\sum_i x_i \leq s$ means the (measurable) set $\{\theta \in \Theta: \sum_i x_i(\theta) \leq s(\theta) \}$ has the same measure as $\Theta$.
Given a ME $(x^*, p^*)$, we often denote the (unique) equilibrium utilities as $u^*_i = \langle v_i, x^*_i \rangle$.
For a finite-dimensional linear Fisher market, it is well known that a ME can be computed via solving the EG convex program. 
Recently, \cite{gao2020infinite} generalized this framework to handle the case of an infinite item space. 
More specifically, consider the following (possibly infinite-dimensional) convex programs.
\begin{align}
    \sup_{x\in (L^\infty_+)^n} \sum_i B_i \log \langle v_i, x_i \rangle\ \ {\rm s.t.}\ \sum_i x_i \leq s. \tag{$\mathcal P_{\rm EG}$} \label{eq:eg-primal}
\end{align}
\begin{align}
    \begin{split}
    \inf_{p\in L^1_+,\, \beta\in \RR^n_+} \left( \langle p, s \rangle - \sum_i B_i \log \beta_i \right) \ \ {\rm s.t.} \ p \geq \beta_i v_i,\ \forall\,i. 
    \end{split} \tag{$\mathcal D_{\rm EG}$}
    \label{eq:eg-dual-beta-p}
\end{align}
The following theorem summarizes the results in \cite[\S 3]{gao2020infinite} regarding the above convex programs capturing market equilibria.
As shown in that work, the above convex programs satisfy strong duality and their optimal solutions (which correspond to ME) can be characterized by the KKT optimality conditions. We slightly generalize the assumptions of \cite{gao2020infinite} by allowing non-uniform item supplies $s$ instead of $s(\theta)=1$ for all $\theta\in \Theta$.
For completeness, a proof, which is mainly based on the proofs of the results in \cite[\S 3]{gao2020infinite}, can be found in the Appendix. 
\begin{theorem}
    The following hold regarding \eqref{eq:eg-primal} and \eqref{eq:eg-dual-beta-p}.
    \begin{itemize}
        \item Both suprema are attained.
        \item Given $x^*$ feasible to \eqref{eq:eg-primal} and $(p^*, \beta^*)$ feasible to \eqref{eq:eg-dual-beta-p}, they are both optimal if and only if the following holds: (i) $\langle p^*, s - \sum_i x^*_i \rangle = 0$ (market clearance), (ii) $ \langle p^* - \beta^*_i v_i, x^*_i \rangle = 0$ (buyer $i$ only receives items within its `winning set' $\{ p^* = \beta^*_i v_i \}$) (ii) and $\langle v_i, x^*_i \rangle = u^*_i := B_i / \beta^*_i$ (buyer $i$ gets its maximum possible utility from $x^*_i$). 
        In this case, $(x^*, p^*)$ is a ME. 
        \item Conversely, for a ME $(x^*, p^*)$, it holds that (i) $x^*$ is an optimal solution of \eqref{eq:eg-primal} and (ii) $(p^*, \beta^*)$, where $\beta^*_i := B_i / \langle v_i, x^*_i \rangle$, is an optimal solution of \eqref{eq:eg-dual-beta-p}.
    \end{itemize}
    \label{eq:thm:eg-capture-me}
\end{theorem}
In the above theorem, $\beta^*_i$ is known as buyer $i$'s \emph{utility price}, i.e., price per unit utility at equilibrium. As is well known, in a ME $(x^*, p^*)$, the allocations $x^*$ are 
\begin{itemize}
    \item Pareto optimal,
    \item envy-free (in a budget weighted sense, i.e., $\langle v_i, x^*_i \rangle / B_i \geq \langle v_i, x^*_k \rangle / B_k$ for all $k\neq i$), 
    \item proportional (i.e., $\langle v_i, x^*_i \rangle \geq \langle v_i, s\rangle / n = 1/n$); see, e.g.,\cite[Theorem 3]{gao2020infinite}.
\end{itemize}
\paragraph{Online Fisher markets and equilibria.} We now consider a simple online variant of the Fisher market setting, referred to as an online Fisher market (OFM). There are $n$ buyers, each with a valuation $v_i \in L^1_+$.
Assume there are discrete time steps $t=1,2, \dots$.
At each time step $t$, an item $\theta_t$ arrives and each buyer $i$ sees a value $v_i(\theta_t)$. The item must be allocated irrevocably to one buyer.
Each buyer $i$ has a budget $B_i > 0$ representing her per-period expenditure rate.
\footnote{This assumption is similar to one made in the literature on budget management in auctions, where each buyers has a per-period expenditure rate and the overall budget equal to the rate times the number of time periods. If a hard budget cap across all time periods is desired, then PACE and similar mechanisms may deplete some buyers' budgets close to the end of the horizon~\citep{balseiro2015repeated,balseiro2017budget,balseiro2019learning}.
}
Next, we introduce the notions of demand, utility level and online market equilibrium in an OFM. All of them are defined based on sequences of arrived items and their prices; they do not require any distributional assumption on the item arrivals. 
\begin{defn} 
    Let the arrived items be $(\theta_\tau)_{\tau\in [t]}$.  
    An allocation (of arrived items) is $(x^\tau_i)_{(\tau,i)\in [t]\times [n]}$, where $x^\tau_i \in [0,1]$ is the fraction of the item $\theta_\tau$ allocated to buyer $i$.\footnote{We allow fractional allocations in the definition for more generality. As we will see, fractional allocation is not needed: PACE generates allocations and prices that satisfy the OME conditions asymptotically via assigning each arrived item to one buyer.}
    Let the prices of the arrived items be $(p^\tau(\theta_\tau))_{\tau \in [t]}$. 
    The \emph{demand} of each buyer $i$ (in hindsight) at time $t$ is
    \begin{align}
        D^t_i = \argmax_{ (z^\tau_i)_{\tau \in [t]}} \left\{ \frac{1}{t}\sum_{\tau=1}^t v_i(\theta_\tau) z^\tau_i: 0 \leq z^\tau_i \leq 1,\, \forall\, \tau,\ \frac{1}{t}\sum_{\tau=1}^t p^\tau(\theta_\tau) z^\tau_i \leq B_i \right\}. \label{eq:def-online-demand-set}
    \end{align}
    Let $\hat{U}^t_i$ be the \emph{utility level} associated with this demand, i.e., the maximum value in \eqref{eq:def-online-demand-set}. 
    An \textbf{online market equilibrium (OME)} is a pair of allocations $(x^\tau_i)_{(\tau,i)\in [t]\times [n]}$ and prices $p^\tau(\theta_\tau)$ such that the following holds. 
    \begin{enumerate}[(i)]
        \item Total allocation does not exceed the unit amount of the item $\sum_i x^\tau_i \leq 1$ for all $\tau$.
        \item Buyers' realized allocations are optimal in hindsight: $(x^\tau_i)_{\tau \in [t]} \in D^t_i$ for all $i$.
        \item Market clearance: $\sum_i x^\tau_i = 1$ for $\tau$ such that $p^\tau(\theta_\tau) > 0$. 
    \end{enumerate}
    \label{defn:ofm-demand-ulevel-ome}
\end{defn}
In words, $\hat{U}^t_i$ is the maximum possible (time-averaged) utility buyer $i$ could have attained via choosing from the arrived items $(\theta_\tau)_{\tau \in [t]}$ in hindsight, subject to their respective posted prices $(p^\tau(\theta_\tau))_{\tau \in [t]}$ and her current total budget $t B_i$, with $D^t_i$ being the set of such utility-maximizing (time-indexed) allocations subject to per-period item availability constraints. An OME is a pair of allocations and prices that make buyers optimal in hindsight and market cleared.

Given an OFM, we define the associated underlying static Fisher market as having the same $n$ buyers and an item space $\Theta$ with supply $s$ being the (unknown) distribution from which the arriving items $\theta_t$ are drawn.
To clarify the concepts of OFM and OME, we consider some simple special cases.
\begin{itemize}
    \item Suppose all item arrivals $\theta_1, \dots, \theta_t$ are known in advance. 
    Then, the OFM is the same as a static $n\times t$ Fisher market with the same buyers and the $t$ items, each having a unit supply. Here, buyer $i$'s valuation of item $\tau$ is $v_{i\tau} = v_i(\theta_\tau)$. To compute an OME, it suffices to solve the classical (finite-dimensional) Eisenberg-Gale convex program, that is, \eqref{eq:eg-primal} with $\Theta = [t]$, $s = (1, \dots, 1) \in \RR^t_+$ and $x\in \RR_+^{n\times t}$. Let the static ME be $(x^*, p^*) \in \RR_+^{n\times t} \times \RR_{++}^t$. When each item $\theta_\tau$ arrives, OME allocates a fraction $x^*_{i\tau}$ of the item to each buyer $i$ and set its price as $p^*_\tau$.
    \item Suppose the sequentially arriving items are drawn i.i.d. from a known underlying distribution $s\in L^\infty_+$ (which specifies a random variable $\theta \sim s$ such that $\PP[\theta\in A] = s(A)$ for any measurable set $A\subseteq \Theta$) and all buyers' valuations $v_i$ are known. 
    Suppose we have also computed a static ME $(x^*, p^*)$ (Definition~
    \ref{defn:me-static}) of a market with buyer valuations $v_i$, budgets $B_i$ and item supplies being the distribution $s$ (the \emph{underlying static market}).
    Then, when a new item $\theta_t$ (which is drawn from the distribution $s$) arrives at time $t$, set its price as $p^*(\theta_t)$ and allocate a fraction $x^*_i(\theta_t) / s(\theta_t)$ of it to each buyer $i$ (assume $s(\theta_t) > 0$, i.e., only items with positive supplies can appear).
    Then, the time-averaged utility of each buyer $i$ is 
    $ \frac{1}{t} \sum_{\tau=1}^t v_i(\theta_t) x^*_i(\theta_t) / s(\theta_t)$,
    which converges to 
    \[ \EE_{\theta \sim s} [ v_i(\theta) x_i(\theta) / s(\theta) ] = \int_\Theta v_i(\theta) x^*_i(\theta) d\theta = u^*_i \ {\rm a.s.}\] by to the Strong Law of Large Numbers.
    Since the online process is carried out using static equilibrium prices and allocations, the static ME properties (Definition~\ref{defn:me-static}) ensure the required OME properties hold asymptotically.
\end{itemize}
The above special cases require full knowledge of either the exact future item arrivals or the underlying static market to attain an OME. 
Next, we propose a simple, distributed dynamics which generates allocations and prices that satisfy the OME conditions asymptotically \textbf{without} requiring such knowledge (in particular, without knowledge of the distribution $s$).

\section{The PACE Dynamics} \label{sec:pace-dynamics}
In this section, we introduce the \textbf{PACE} (Pacing According to Current Estimated utility) dynamics that prices and allocates sequentially arriving items via (i) maintaining a \emph{pacing multiplier} for each buyer and (ii) simple, distributed updates.\footnote{Pacing and pacing multipliers are terminology in budget management in large-scale ad auctions \cite{conitzer2018multiplicative,conitzer2019pacing}.} 
In \S\ref{sec:conv-analysis-pace}, we will show that PACE is an instantiation of dual averaging \cite{xiao2010dual}, a stochastic first-order method for regularized optimization, applied to a reformulation of \eqref{eq:eg-dual-beta-p}.

In the PACE dynamics, each buyer maintains a pacing multiplier $\beta^t_i$ (with an initial value $\beta^0_i = (B_i + 1)/2$, or any value in $[B_i, 1]$). At time step $t$, the following events take place.
\begin{enumerate}[(a)]
\item An item $\theta_t$ appears and each buyer $i$ sees a value $v_i(\theta_t)$ for the item.
\item \label{item:pace-buyer-bid-step} Each buyer $i$ bids their paced value $\beta^t_i v_i(\theta_t)$ for the item.
\item \label{item:pace-pricing-and-pay-step} The item is allocated to the highest bidder (the \emph{winner} at $t$): 
$i_t = \argmax_i \beta^t_i v_i(\theta_t)$, 
with ties broken arbitrarily. 
For concreteness, we always choose the lowest winning index, i.e., 
\[ i_t = \min \argmax_i \beta^t_i v_i(\theta_t).\] 
Then, the price of $\theta_t$ is set by the first-price rule
\[ p^t(\theta_t) = \max_i \beta^t_i v_i(\theta_t) = \beta^t_{i_t} v_i(\theta_t)\] 
and the winner $i_t$ pays this price $p^t(\theta_t)$ for the item $\theta_t$.
\item Each buyer $i$ gets a utility 
\[ u^t_i = v_i(\theta_t) \ind\{ i = i_t \}. \] 
In other words, the winner $i_t$ gets $v_{i_t}(\theta_t)$ and other buyers get zero. 
\item Each buyer $i$ updates its cumulative average utility $\bar{u}^t_i$:
 \[\bar{u}^t_i = \frac{1}{t}\sum_{\tau=1}^t u^\tau_i = \frac{t-1}{t}\bar{u}^{t-1}_i + \frac{1}{t}u^t_i. \]
\item Each buyer $i$ updates their pacing multiplier $\beta^{t+1}_i$ as follows:
\[ \beta^{t+1}_i = \Pi_{ [l_i, h_i] } ( B_i / \bar{u}^t_i) := \min \{ \max\{l_i, B_i / \bar{u}^t_i \}, h_i \}.\] 
where $l_i = B_i / (1+\delta_0)$ and $h_i = 1+\delta_0$ for some fixed $\delta_0 > 0$ (e.g., $\delta_0 = 0.05$). \label{item:pace-update-beta-step}  
\end{enumerate}
As will be seen in \S \ref{sec:conv-analysis-pace}, buyer $i$'s equilibrium pacing multiplier (utility price) satisfies $l_i < \beta^*_i < h_i$ and her per-period utility $u^t_i$ corresponds to the $i$th component of a stochastic subgradient of a function on $\beta$ in a reformulation of the convex program \eqref{eq:eg-dual-beta-p}, on which we run dual averaging.
Furthermore, the update rule for $\beta_i^{t+1}$ is such that, if the realized utilities $\bar{u}^t_i$ were the true static equilibrium utility for buyer $i$, then $\beta_i^{t+1}$ would be the equilibrium multiplier. Note that PACE does not randomize (any randomness can only come from the market environment from which item arrivals are drawn) and assigns every item to a single buyer without splitting it. 

The simplicity and distributed nature of PACE makes it desirable for large-scale practical use. 
\begin{itemize}
    \item It can be run on arbitrary sequential item arrivals and only requires buyers' valuations $v_i(\theta_t)$ on the arrived items (rather than all valuations $v_i$ over the potentially large item space). 
    No parameter tuning is needed (in particular, no stepsize tuning as in many first-order optimization methods).
    \item When run as a \emph{centralized} allocation mechanism, PACE only needs to maintain $O(n)$ scalars, namely, $\beta^t_i$, $B_i$ and $\bar{u}^t_i$ for all $i$. At time $t$, it observes buyers' valuations $v_i(\theta_t)$ of the item $\theta_t$, compute bids $\beta^t_i v_i(\theta_t)$, finds the winner $i_t$, set the price as the maximal bid $\beta^t_{i_t} v_{i_t}(\theta_t)$ and allocates the item to the winner; finally, it updates $\bar{u}^t$ and $\beta^{t+1}$ as in \ref{item:pace-update-beta-step}, which takes $O(n)$ time.
    \item PACE can also be run among the buyers in a \emph{decentralized} manner, in which case each buyer only maintains two scalar values: the pacing multiplier $\beta^t_i$ and time-averaged utility $\bar{u}^t_i$.
    When a new item arrives, each buyer only performs a few simple arithmetic operations to create a bid $\beta^t_i v_i(\theta_t)$, receives her utility (if she wins) and subsequently updates $\bar{u}^t_i$ and $\beta^{t+1}_i$. 
\end{itemize}
These make PACE suitable for Internet-scale online fair division and online Fisher market applications. 
In particular, it is very reminiscent of how Internet advertising auctions are run. 
There, a similar auction-based system is used, with the pacing multiplier ensuring that each advertiser smooths out their budget expenditure across the many auctions. 
The primary difference between this and our setting is that (i) the auction can be first-price or second-price and (ii) buyers usually have \emph{quasilinear} utilities, that is, utility of the item minus the expenditure (price paid) \cite{conitzer2018multiplicative,balseiro2015repeated,balseiro2017budget,balseiro2019learning}. 
In \S\ref{app:ql}, we extend PACE to quasilinear utilities, which provides a novel online algorithm for first-price pacing equilibrium computation \cite{conitzer2019pacing}.
\section{Dual Averaging} \label{sec:da-general}
In this section, we briefly recap the setup and general convergence results of \emph{dual averaging} \cite{xiao2010dual,nesterov2009primal}, which will be used in the analysis of PACE. 
First, we introduce some notation for this and subsequent sections. Let $\mathbf{e}^{(i)}$ denote the $i$'th unit basis vector in $\RR^n$ and $\mathbf{1}\in \RR^n$ denote the vector of $1$'s. 
For $x, y\in \RR^n$, $[x,y]$ denotes the Cartesian product of intervals $\prod_{i=1}^n [x_k, y_k] \subseteq \RR^n$. 
All norms $\|\cdot \|$ without a subscript are Euclidean $2$-norms, unless otherwise stated. 

Let $\Psi$ be a closed convex function with domain $\dom \Psi := \{ w\in \RR^n: \Psi(w)<\infty \}$. Let $Z \subseteq \RR^d$ be an arbitrary sample space.
For each $z\in Z$, let $f_z$ be a convex and subdifferentiable function on $\dom \Psi$. Considers the following regularized convex optimization problem \cite[\S 1.1]{xiao2010dual}:
\begin{align}
    \min_w \EE f_z(w) + \Psi(w), \label{eq:da-stoch-opt-std-form}
\end{align}
where the expectation is taken over a probability distribution $\mathcal{D}$ on $Z$. 
A more general online optimization setting, as described in \cite[\S1.2]{xiao2010dual} is as follows.
At each time $t=1, 2, 3, \dots$, we must choose an action $w^t$ before a new, unknown convex loss function $f_t$ arrives, which incurs a loss $f_t(w^t)$ (a special case is i.i.d. sampled functions, i.e., $f_t = f_{z_t}$, where $z_t$ are i.i.d. samples drawn from $\mathcal{D}$). 
The goal is to minimize \emph{regret} when comparing our sequence of actions $w^1,w^2,\ldots$ to any \emph{fixed} action $w$. 
Here, the regret against $w$ is defined as 
\[ R_t(w) := \sum_{\tau=1}^t \left( f_\tau(w^\tau) + \Psi(w^\tau) \right) - \sum_{\tau=1}^t (f_\tau(w) + \Psi(w))\] 
and the overall (maximal) regret up to time $t$ is $R_t = \max_w R_t(w)$.
We assume access to an oracle that, given any $f_t$ and $w\in \dom \Psi$, returns a subgradient $g^t \in \partial f_t(w)$. 
The dual averaging algorithm (DA) \cite[Algorithm 1]{xiao2010dual} is as follows. First, set $w_1 \in \dom \Psi$ and $\bar{g}^0 = 0$. Then, for each $t=1,2,\dots$, DA performs the following steps:
\begin{enumerate}[(1)]
    \item Observe $f_t$ and compute $g^t \in \partial f_t(w^t)$. \label{item:da-compute-g(t)}
    \item Update the average subgradient (the \emph{dual average}) via $\bar{g}^t = \frac{t-1}{t}\bar{g}^{t-1} + \frac{1}{t}\bar{g}^t$. \label{item:da-update-g-bar}
    \item Compute the next iterate $w^{t+1} = \argmin_w \{ \langle \bar{g}^t, w\rangle + \Psi(w) \}$. \label{item:da-update-w}
\end{enumerate}
Here, we do not employ any auxiliary regularizing function, since our problem has a natural source of strong convexity (i.e., a strongly convex $\Psi$) through the $-B_i \log \beta_i$ terms in \eqref{eq:eg-dual-beta-p}.
The following convergence guarantee on DA is proved as part of the proof of Corollary 4 in \cite{xiao2010dual}.
\begin{theorem} \label{thm:general-conv}
    Dual averaging generates iterates $w^t$ such that 
    \[ \EE \|w^t - w^*\|^2 \leq \frac{(6+\log t) G^2}{t \sigma^2}, \]
    where $G^2$ is an upper bound on $\EE \|g^t\|^2$, $t=1,2,\dots$ and $\sigma$ is the strong convexity modulus of $\Psi$.
\end{theorem}
When solving the stochastic optimization problem \eqref{eq:da-stoch-opt-std-form}, 
in Theorem \ref{thm:general-conv}, we can set $G^2$ to be an upper bound on $ \sup_{w\in \dom \Psi} \EE \|g_z(w)\|^2$, where $g_z(w)$ is a subgradient oracle mapping each $(z,w)\in Z\times \dom \Psi$ to a subgradient and the expectation is over $z\sim \mathcal{D}$ and possible randomness of the subgradient oracle.
We will shortly see that a reformulation of \ref{eq:eg-dual-beta-p}, when cast into the form \eqref{eq:da-stoch-opt-std-form}, exhibits stochastic subgradients that are exactly buyers' received utilities in each time step. Using Theorem \ref{thm:general-conv}, we can show that the sequence of pacing multipliers $\beta^t$ generated by PACE converges to the underlying (equilibrium) utility prices $\beta^*$ of the static Fisher market. 

\section{Convergence analysis of the PACE dynamics} \label{sec:conv-analysis-pace}
We will now show that PACE correspond to running DA on the vector $\beta^t$ of pacing multipliers for the buyers. To this end, we first reformulate \eqref{eq:eg-dual-beta-p} into a (finite-dimensional) convex program in $\beta$ in the form of \eqref{eq:da-stoch-opt-std-form}:
\begin{align}
    \min_\beta \  \left( \langle \max_i \beta_i v_i, s\rangle - \sum_i B_i \log \beta_i\right)  \ {\rm s.t.}\ \beta \in [B/(1+\delta_0), (1+\delta_0)\ones],
    \label{eq:eg-dual-beta-bounds}
\end{align}
where $\delta_0 >0$ is an arbitrarily small constant. The bounds on $\beta$ do not change the optimal solution, because $\beta^*_i \in (B_i, 1)$ for each $i$.
Detailed steps of the reformulation are given in Appendix \ref{app:proofs-and-deriv}.

In order to run DA, we need to compute a subgradient of $f_\theta: \beta\mapsto \max_i \beta_i v_i(\theta)$ at any $\theta\in \Theta$. Following \cite[\S 5]{gao2020infinite}, since $f_\theta$ is a piecewise linear function, a subgradient is 
\[ g_\theta(\beta) := v_{i^*}(\theta) \mathbf{e}^{(i^*)} \in \partial f_\theta(\beta),\] 
where $i^* = \min \argmax_i \beta_i v_i(\theta)$ is the winner (see, e.g., \cite[Theorem 3.50]{beck2017first}).


We can now show that the PACE dynamics corresponds to running DA on \eqref{eq:eg-dual-beta-bounds}.
First, choose an arbitrary $\beta^0 \in (B/(1+\delta_0), (1+\delta_0) \ones)$ and $\bar{g}^0 = 0$. At each time step $t=1, 2, \dots$, given the current pacing multiplier $\beta^t$, DA applied to \eqref{eq:eg-dual-beta-bounds} unrolls the following steps. 
\begin{itemize}
    \item An item $\theta_t$ arrives, having value $v_i(\theta_t)$ for each buyer $i$. The function $f_t$ in DA is 
    \[ f_{\theta_t}: \beta \mapsto \max_i \beta_i v_i(\theta_t).\]
    \item The winner is $i_t = \min \argmax_i \beta^t_i v_i(\theta_t)$ and a  subgradient is 
    \[ g^t = v_{i_t j_t} \mathbf{e}^{(i_t)} \in \partial f_t(\beta^t).\] 
    Its $i$th entry is exactly the realized (single-period) utility of individual $i$ at time $t$ in PACE, that is, $g^t_i = v_i(\theta_t) \ind\{ i=i_t \} = u^t_i$.
    \item Update the dual average (time-averaged utilities): for each $i$, compute $\bar{g}^t = \frac{t-1}{t}\bar{g}^{t-1} + \frac{1}{t}g^t$, i.e., 
    \[ \bar{g}^t_i = \frac{t-1}{t} \bar{g}^{t-1}_i + \frac{1}{t} v_i(\theta_t) \ind\{ i=i_t \}.\]
    \item Update the pacing multipliers:
    \[ \beta^{t+1} = \argmin_{\beta\in [B/(1+\delta_0), (1+\delta_0)\ones]} \left\{ \langle \bar{g}^t, \beta\rangle - \sum_i B_i \log \beta_i \right\}.\]
    The minimization problem is separable in each $i$ and exhibits a simple and explicit solution which recovers step \ref{item:pace-update-beta-step} in PACE (where $\bar{g}^t_i = \bar{u}^t_i$): 
    \[ \beta^{t+1}_i = \argmin_{\beta_i \in  [B/(1+\delta_0), 1+\delta_0]} \left\{\bar{g}^t_i \beta_i - B_i \log \beta_i \right\}  \ \Rightarrow\ \beta^{t+1}_i = \Pi_{ [B_i/(1+\delta_0), 1+\delta_0] } \left(\frac{B_i}{\bar{u}^t_i}\right). \]
\end{itemize}
As mentioned earlier, PACE does not require a stepsize parameter. 
This is because DA is stepsize-free given a strongly convex regularizer $\Psi$, which is indeed the case in our reformulation~\eqref{eq:eg-dual-beta-bounds}.
In addition, in the above update step for $\beta^{t+1}_i$, the directions of change are as follows.
\begin{itemize}
    \item For a non-winner $i\neq i_t$, we have $u^t_i=0$ and hence $\bar{u}^t_i \leq \bar{u}^{t-1}_i$. This implies $\beta^{t+1}_i \geq \beta^t_i$. In words, a non-winner's pacing multiplier weakly increases. The increase is strict if $\bar{u}_i^{t-1} >0$, i.e., buyer $i$ has already received a nonzero utility.
    \item For the winner $i_t$, $\bar{u}^t_{i_t}$ may become greater than $\bar{g}^{t-1}_{i_t}$, in which case $\beta^{t+1}_{i_t} \leq \beta^t_{i_t}$. In words, the winner's pacing multiplier may go up or down.
\end{itemize}

In order to analyze PACE, we assume $v_i(\Theta) = 1$, $v_i \in L^\infty_+$ (normalized and a.e.-bounded valuations)\footnote{The a.e.-boundedness assumption is needed in subsequent convergence analysis. Since $\Theta$ has a finite measure, it holds that $L^\infty_+ \subseteq L^1_+$. For a finite item space $\Theta=[m]$, both are equal to $\RR^m_+$.} and that there is an underlying item distribution $s\in L^\infty_+$ from which the item arrivals $\theta_t$, $t=1,2,\dots$ are drawn i.i.d.\footnote{The distributional assumption on item arrivals (i.e., they are drawn i.i.d. from an unknown distribution $s$) is needed to establish asymptotic equilibrium properties of PACE. See Appendix \ref{app:proofs-and-deriv} for an example that any algorithm can yield arbitrarily suboptimal allocations without such a distributional assumption.}
Define the underlying static Fisher market as one having the same $n$ buyers (each with valuation $v_i$ and budget $B_i$) and item supplies $s$. Denote the equilibrium utilities and utility prices w.r.t. the underlying static market as $u^*$ and $\beta^*$, respectively.
We further assume that the valuations are $v_i \in L^\infty_+$ (i.e., a.e.-bounded on the item space). This is not restrictive: since an individual item $\theta$ has value $v_i(\theta)$ for each buyer $i$, it should be a finite value.

\paragraph{Convergence of pacing multipliers.}
After aligning PACE with DA, the convergence of the pacing multipliers $\beta^t$ follows directly from Theorem~\ref{thm:general-conv}.
\begin{theorem} 
    PACE generates pacing multipliers $\beta^t$, $t=1,2,\dots$ such that
    \[\EE \|\beta^t - \beta^*\|^2 \leq \frac{(6+\log t)G^2}{t\sigma^2}, \]
    where $G^2 = \max_i \EE_{\theta\sim s} [v_i(\theta)^2] \leq \max_i \|v_i\|_\infty^2$, $\sigma = \frac{\min_i B_i}{(1+\delta_0)^2}$.
    \label{thm:conv-beta(t)}
\end{theorem}
In other words, we have mean-square convergence of $\beta^t$ to $\beta^*$ at a $O((\log t)/t)$ rate. 
Since $\|B\|_1 = 1$, we have $ \min_i B_i \leq 1/n$. 
Hence, $\sigma = O(1/n)$ and the constant in the bound is $\Omega( n^2 )$. 
Whether such dependence on $n$ can be improved via new analysis remains an interesting research question.



\paragraph{Convergence of utilities.}
We next show that the time-averaged utility $\bar{u}^t$ (which is equal to the dual average $\bar{g}^t$) converges to the equilibrium utility vector $u^*$ of the underlying Fisher market. 
A key step in the proof is to bound the probability of a projection in updating $\beta^{t+1}_i$, that is, $\PP[ B_i / \bar{u}^t_i \notin [l_i, u_i]]$. 
\begin{theorem} 
    For each $i$, let $\epsilon_i := \min \{ h_i - \beta^*_i, \beta^*_i - l_i \}>0$ be the minimum distance to the endpoints of the pacing-multiplier interval and  $\|v\|_\infty := \max_i \|v_i \|_\infty$. 
    It holds that
    \[ \EE (\bar{u}^t_i - u^*_i)^2 \leq \left( \frac{\|v_i\|_\infty^2}{\epsilon_i^2} + \left(\frac{1+\delta_0}{B_i}\right)^2 \right) \EE (\beta^{t+1}_i - \beta^*_i)^2. \]
    Hence, letting $C = \frac{1}{(\min_i B_i)^2} \left(( \|v\|_\infty / \delta_0 )^2 + (1+\delta_0)^2\right)$, we have
    \[ \EE \|\bar{u}^t - u^* \|^2 \leq C\cdot \frac{(6 + \log (t+1)) G^2}{(t+1) \sigma^2}. \]
    \label{thm:conv-utilities}
\end{theorem}

Note that $C = \Omega(n^2 )$. Hence, in this and the next theorems, the constant in the bound is $\Omega(n^4)$, which arises from $C$ and $\sigma = O(1/n)$. 


\paragraph{Convergence of expenditures.} The \emph{expenditure} of buyer $i$ at time step $t$ is
\[ b^t_i = \beta^t_i v_i(\theta_t) \ind\{ i = i_\tau \}.\] 
In other words, only the winner $i_t$ spends a nonzero amount, which is its bid.
Let $\bar{b}^t_i = \frac{1}{t}\sum_{\tau = 1}^t b^\tau_i$ be buyer $i$'s average expenditure. 
Utilizing the above convergence results, 
we show mean-squared convergence of $\bar{b}^t$ to $B$ at an $O((\log t)^2/t)$ rate.
\begin{theorem}
    For each $i$, it holds that 
    \[ \EE (\bar{b}^t_i - B_i) \leq 2\left[ (\beta^*_i)^2 \EE (\bar{g}^t_i - u^*_i)^2 + 2\|v_i\|_\infty^2 \frac{1}{t}\sum_{\tau=1}^t \EE(\beta^\tau_i -\beta^*_i)^2 \right]. \]
    For $t\geq 3$ and the constant $C$ defined in Theorem~\ref{thm:conv-utilities}, we have
    \[ \EE\|\bar{b}^t - B\|^2 \leq \frac{2 G^2}{t\sigma^2} \left( 6(C+\|v\|_\infty^2) + (C + 6 \|v\|_\infty^2) \log t + \frac{\|v\|_\infty^2}{2} (\log t)^2 \right). \]
    \label{thm:conv-expenditures}
\end{theorem}
\paragraph{PACE attains OME asymptotically. } 
Next, we show that PACE attains OME asymptotically, i.e., it generates allocations and prices that make buyers \emph{no-regret} and \emph{envy-free} in the limit (these notions will be clarified shortly).
Let $x^t_i := \ind\{ i = i_t \}$ denote whether buyer $i_t$ is the winner (i.e., whether she is allocated the item $\theta_t$ at time $t$)
Utilizing Theorems~\ref{thm:conv-utilities}~and~\ref{thm:conv-expenditures}, we can show that buyer $i$'s regret, that is, the difference between the maximum possible utility in hindsight $\hat{U}^t_i$ (Definition~\ref{defn:ofm-demand-ulevel-ome}) and the realized utility $\bar{u}^t_i$, vanishes as $t$ grows. The same holds for each buyer's envy. In other words, at a large $t$, in hindsight, no buyer prefers another buyer's set of allocated items (up to a vanishing error).\footnote{In a static market, given an allocation $x\in \RR^{n\times m}_+$, the (maximum, budget-weighted) \emph{envy} of buyer $i$ toward others' bundles is $\rho_i(x) = \max_k \langle v_i, x_k\rangle / B_k - \langle v_i, x_i \rangle / B_i$ (see, e.g., \citep{varian1974equity,budish2011combinatorial}). 
It is well-known that $\rho_i(x^*) = 0$ for all $i$ at equilibrium, a consequence of buyer optimality (Definition~\ref{defn:me-static}).}
\begin{theorem}
    Denote 
    \[ \xi^t_i = |\bar{u}^t_i - u^*_i|, \ \Delta^t_i = | \bar{b}^t_i - B_i |, \ \gamma_t = \frac{\|v\|_\infty}{t} \sum_{\tau=1}^t \|\beta^\tau - \beta^*\|_\infty.\] 
    Let $r^t_i := \max\{\hat{U}^t_i - \bar{u}^t_i , 0\}$ be the \emph{regret} of buyer $i$ at time $t$. 
    Then, it holds that
    \[ r^t_i \leq \xi^t_i + \gamma_t / B_i,\  \EE (r^t_i)^2 = O\left((\log t)^2 / t \right).\]
    Furthermore, let the \emph{envy} of buyer $i$ (w.r.t. all other buyers) at time $t$ be 
    \[ \rho^t_i = \max_k \bar{u}^t_{ik} / B_k - \bar{u}^t_i/B_i,\] 
    where 
    $ \bar{u}^t_{ik} = \frac{1}{t}\sum_{\tau=1}^t v_i(\theta_\tau) x^\tau_k$ is buyer $i$'s time-averaged utility given her own valuations and of buyer $k$'s allocations.
    Denote 
    \[ \eta^t_i = \frac{1}{t}\sum_{\tau=1}^t (p^*(\theta_t) -\beta^\tau_i v_i(\theta_t))x^t_i.\] 
    It holds that
    \begin{align*}
        \rho^t_i \leq \frac{1}{B_i} \left( \xi^t_i + \max_{k\neq i} \frac{\Delta^t_k + \eta^t_k}{B_k} \right) \ \ \text{and} \ \
        \EE (\eta^t_i)^2 \leq \frac{\|v\|_\infty^2 G^2}{ t \sigma^2} \left( 6(1+\log t) + \frac{(\log t)^2}{2} \right).
   \end{align*}
   Hence, $\EE(\rho^t_i)^2 = O\left( (\log t)^2/t\right)$. 

    \label{thm:pace-conv-to-OME-no-regret}
\end{theorem}
In light of Definition~\ref{defn:ofm-demand-ulevel-ome}, Theorem~\ref{thm:pace-conv-to-OME-no-regret} shows that $(x^\tau_i)_{(i,\tau)\in [n]\times [t]}$ is approximately optimal for buyer $i$. Since PACE also clears the market, we conclude that it attains OME asymptotically.
Recall that theorem~\ref{thm:conv-utilities} ensures that buyers' $\bar{u}^t_i$ converge to their static equilibrium utilities $u^*_i$. 
Since the latter satisfy Pareto optimality and proportional share guarantee ($u^*_i \geq B_i$ for all $i$), so are the time-averaged realized utilities in the limit. Together with Theorem~\ref{thm:pace-conv-to-OME-no-regret}, we conclude that PACE achieves the said fairness and efficiency guarantees, namely, Pareto optimality, envy-freeness and proportional-share guarantee, asymptotically.

\section{Experiments} \label{sec:experiments}

We evaluate the PACE dynamics in several real and synthetic datasets, namely, MovieLens, Household Items and an infinite-dimensional market instance with item space $\Theta = [0,1]$ and $v_i$ being linear functions on $[0,1]$. 
For the first two datasets, see \cite{kroer2019computing} for more information and exploratory data analysis. 
For all datasets, we consider the CEEI (fair division) setting where $B_i = 1/n$ for all $i$. For each dataset (with number of buyers $n=1500,2876,100$, respectively), we run PACE for $T=10n$ time steps (iterations).
More details on the experiments and additional plots displaying convergence of expenditures can be found in Appendix~\ref{app:experiment-details}.
Figure~\ref{fig:plot-3-datasets} displays the mean values of the average and maximum relative errors of the pacing multipliers and time-averaged cumulative utilities over $10$ repeated experiments with different seeds (relative errors of cumulative spending w.r.t. total budgets are plotted separately in Appendix~\ref{app:experiment-details}). 
The standard errors are also displayed as vertical bars but are very small and nearly invisible.
Vertical dotted lines indicate $t=10n$ The figures do not show the initial iterates $t=1, \dots, 5n$.

We see that PACE converges very quickly numerically: within $10$ epochs ($10n$ time steps) average deviations in most quantities falls within $5 \%$ of the equilibrium quantity, with the worst case not far behind. An important point is that budget spend takes much longer to converge than utility. This demonstrates an important practical difference for using PACE in an allocation scenario where budgets are `real money' (e.g. Internet ad impressions) as compared to a CEEI-like setting, where budgets are faux currency only used for fair division.

\section{Conclusion}
We introduced the concept of an online Fisher market and proposed the PACE dynamics. We showed that when items arrive sequentially and stochastically, PACE converges to equilibrium outcomes of the underlying market model. 
Furthermore, we showed that, as a consequence of this, PACE can be used in online fair division problems to generate an online allocation that, asymptotically, achieves the compelling fairness properties of CEEI.

Many questions remain for future research. We mostly focused on the case where budgets are faux currency and there are many open questions for adapting PACE to a real-money budget-management setting as well as more complicated nonlinear utility models. Another imperative question, especially for practitioners, is whether PACE guarantees some level of incentive-compatibility.

\begin{center}
    \begin{figure}
        \includegraphics[scale=0.30]{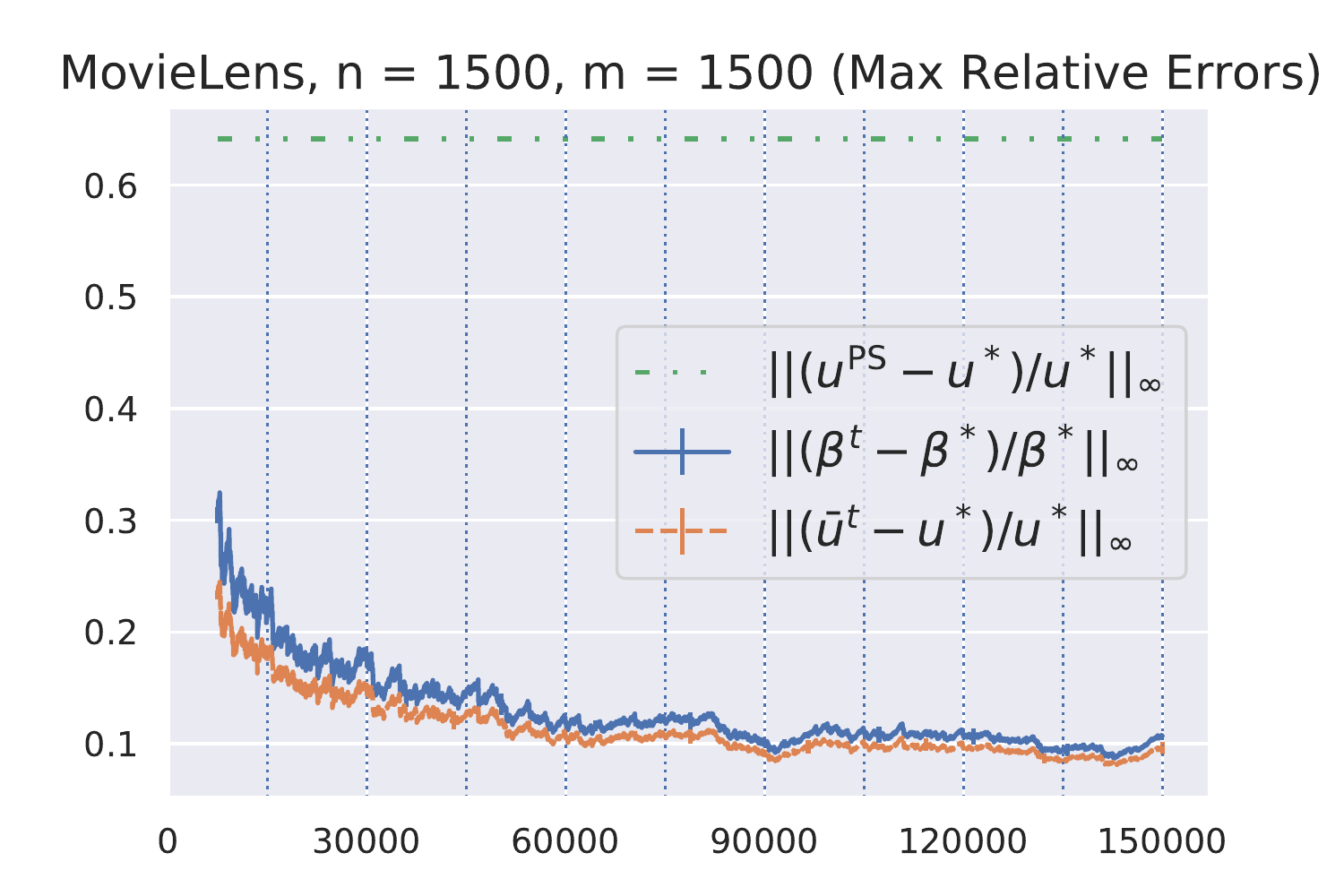} \includegraphics[scale=0.30]{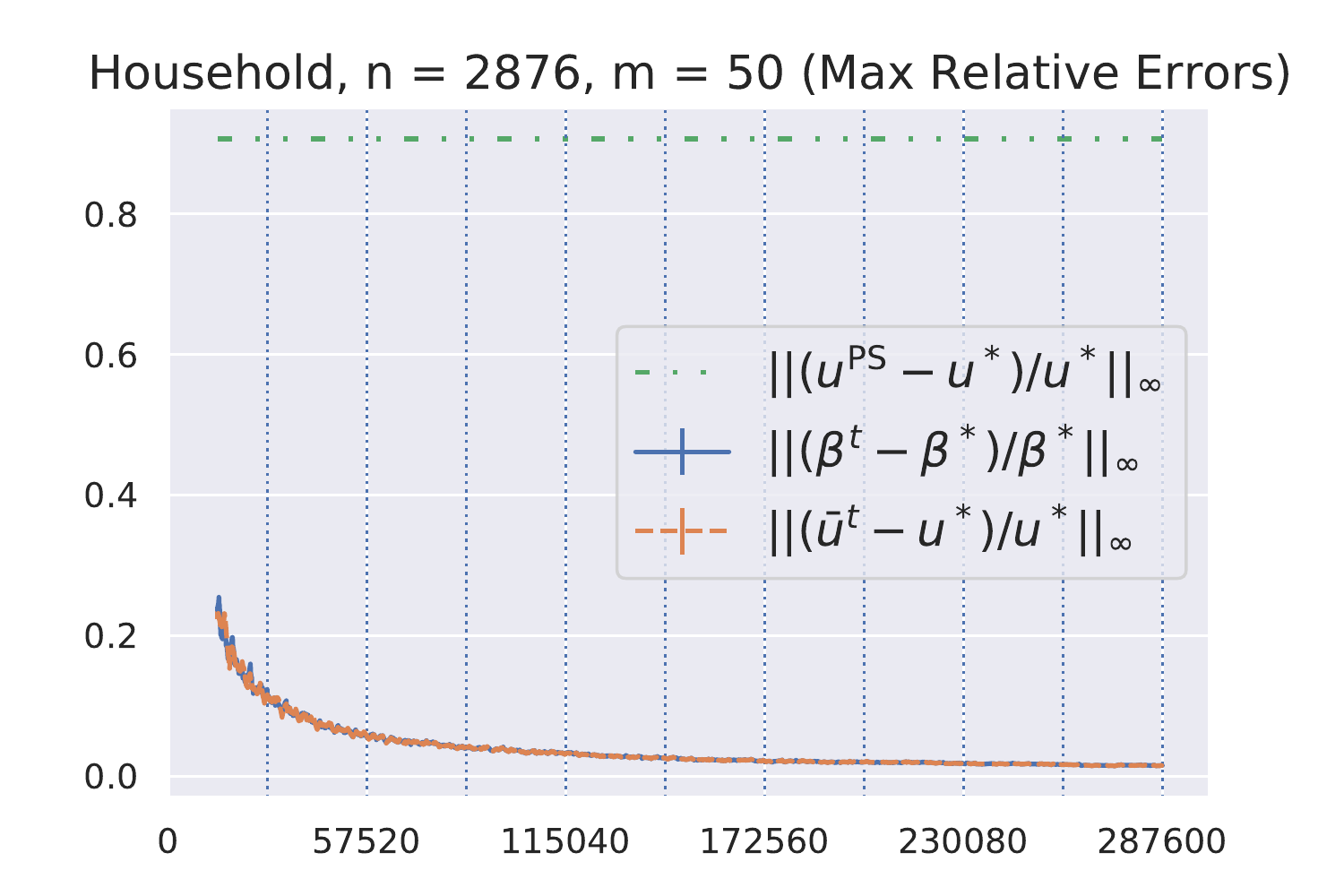} \includegraphics[scale=0.30]{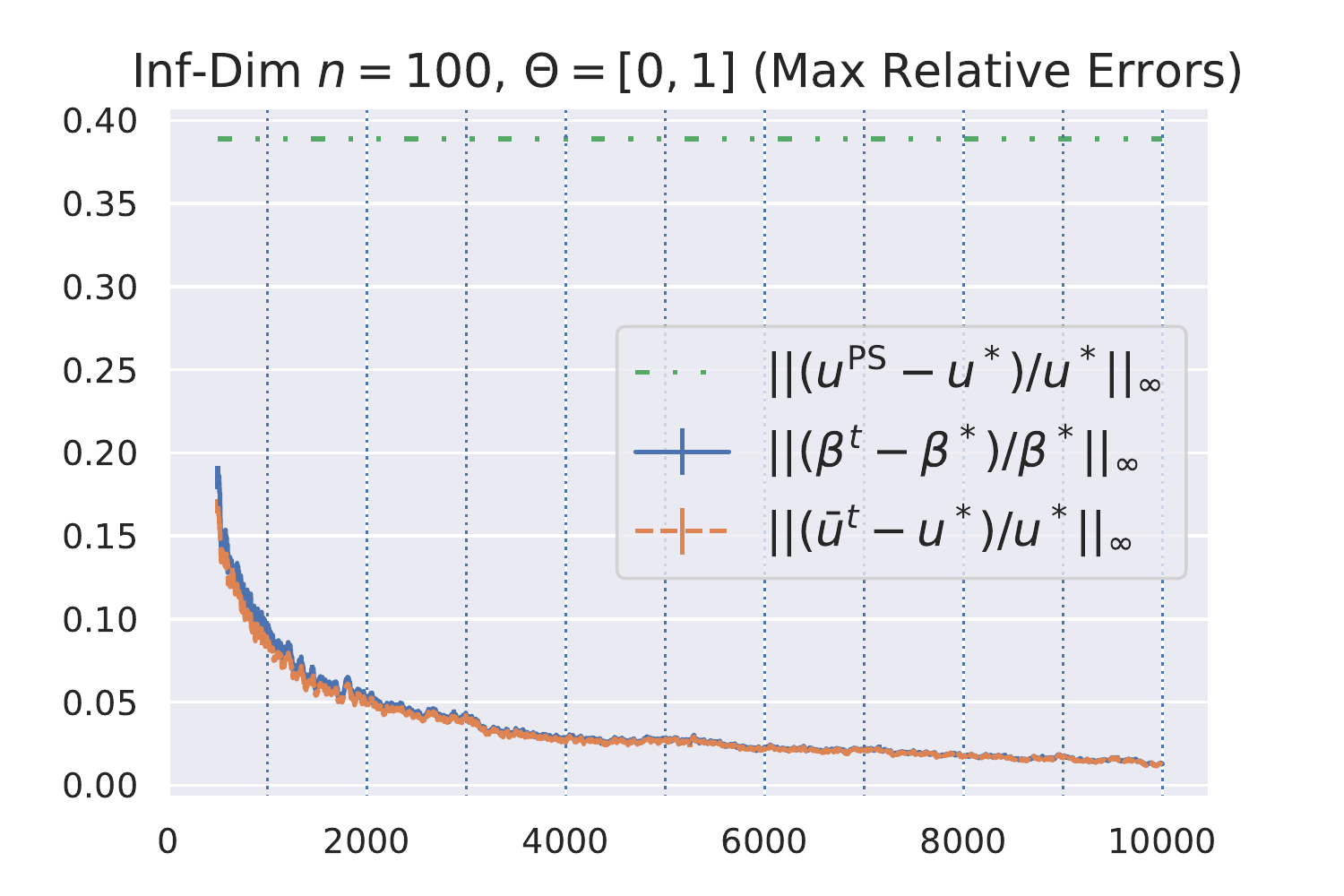} 

        \includegraphics[scale=0.30]{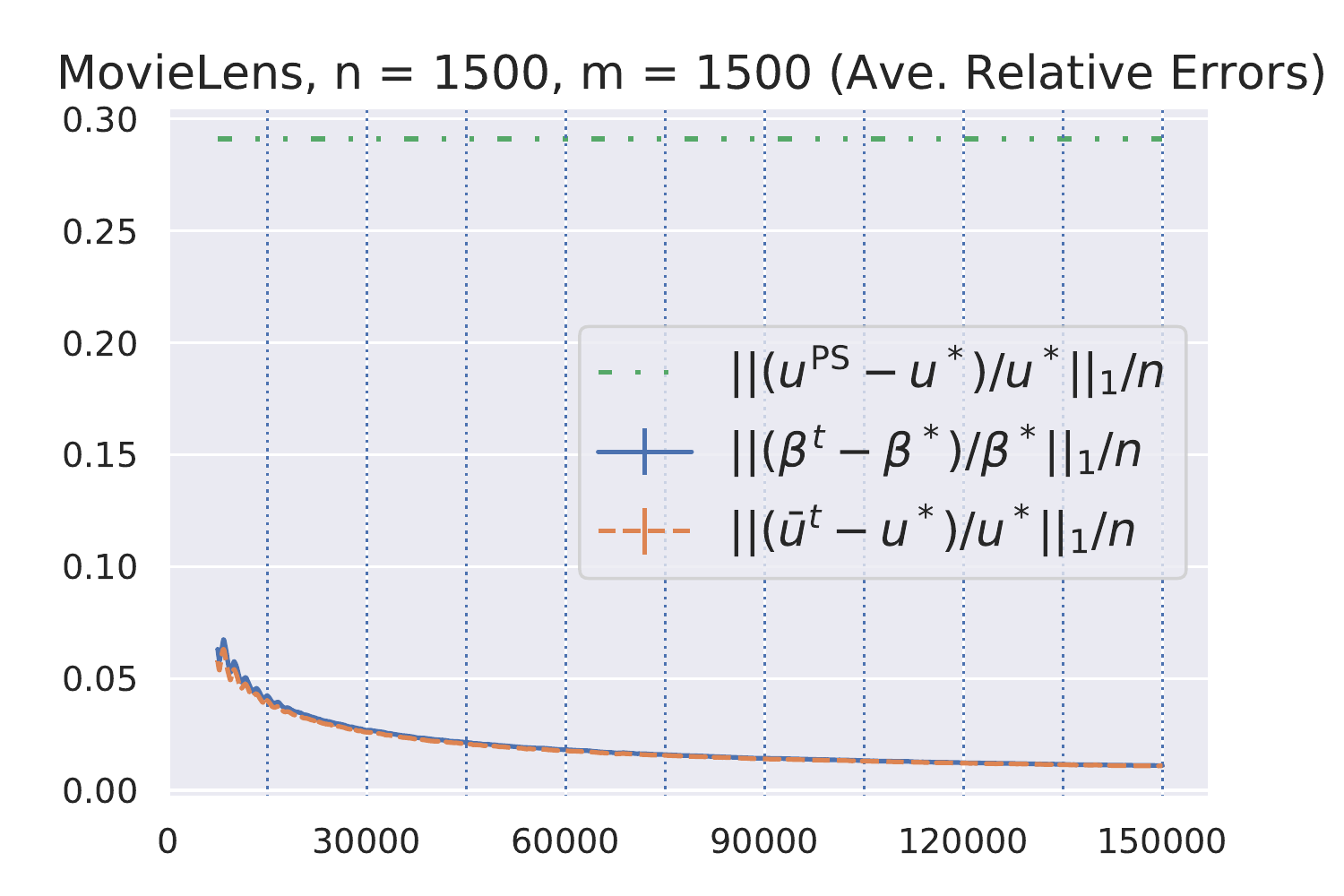} \includegraphics[scale=0.30]{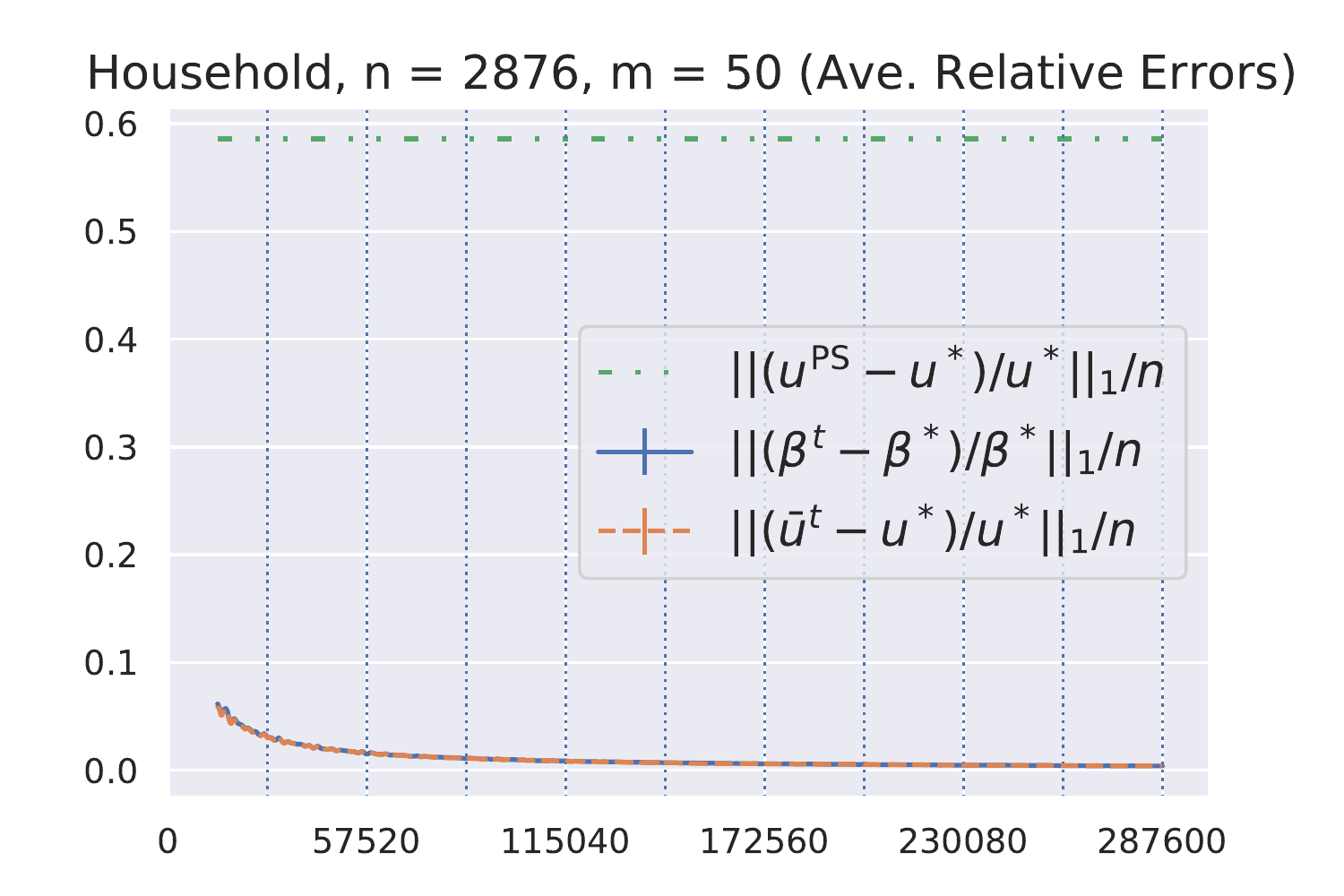} \includegraphics[scale=0.30]{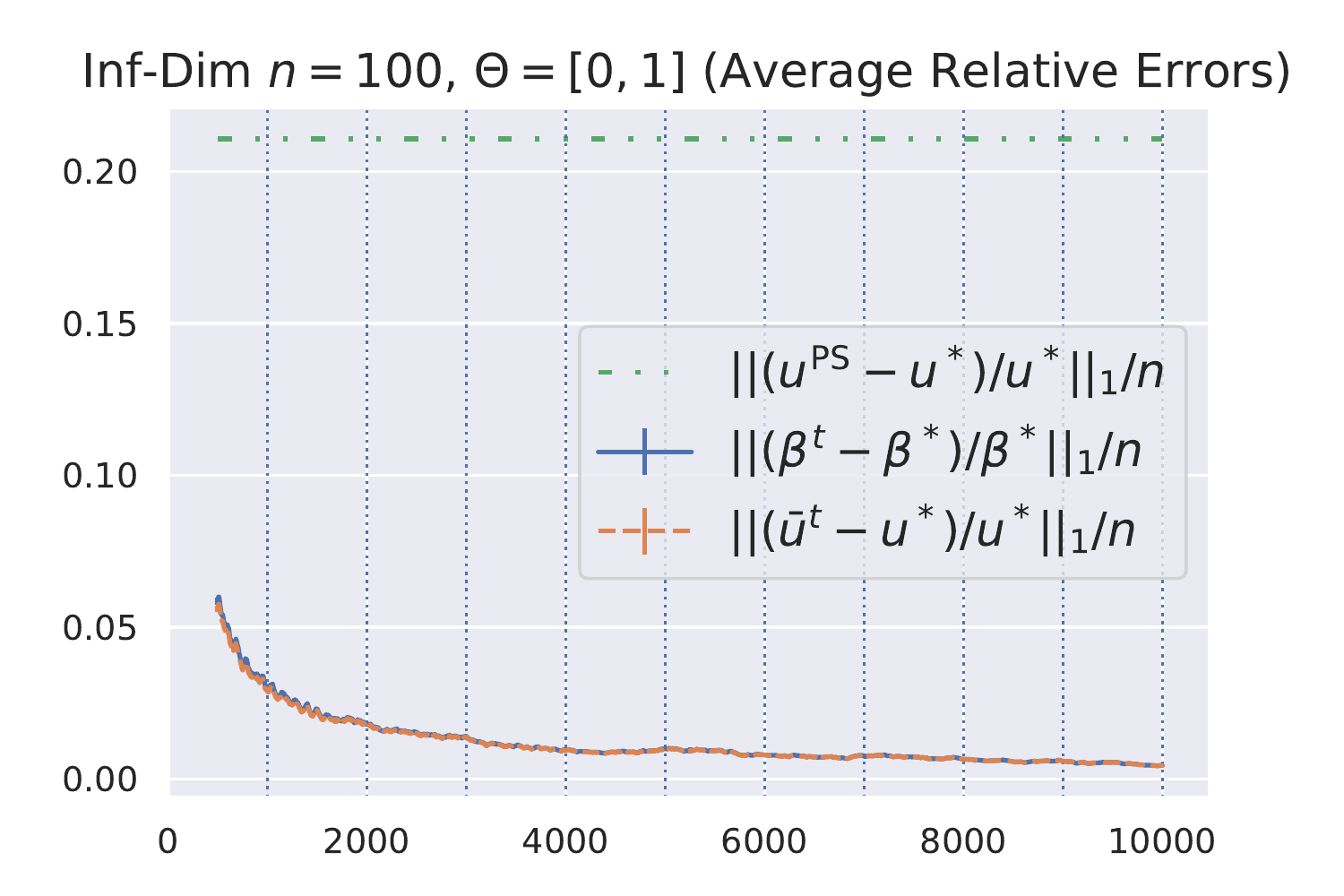}
        \caption{In all of our markets, iterates of the PACE dynamics quickly converges to their static equilibrium values both in the average case and the worst-off-buyer case. 
        The horizontal line shows the fraction of $u^*$ achieved by the proportional share solution.
        The PACE utilities quickly outperform the proportional share utilities. Vertical lines indicate when $t$ is a multiple of $10 n$. }
        \label{fig:plot-3-datasets}
    \end{figure}
\end{center}

\newpage
\bibliographystyle{abbrvnat}
\bibliography{refs.bib}

\newpage
\appendix
\section{Related Work} \label{app:related-work}

The problem of market equilibrium computation has been of interest in economics for a long time (see, e.g., \citet{nisan2007algorithmic}). There is a large literature focusing on computation of equilibrium in the specific case of (finite-dimensional) Fisher markets through various convex optimization formulations~\citep{eisenberg1959consensus,shmyrev2009algorithm,cole2017convex,kroer2019computing} and gradient-based methods~\citep{birnbaum2011distributed,nesterov2018computation,gao2020first}. Other works extend these results to settings such as quasilinear utilities, capped utilities, indivisible items, or imperfectly specified utility functions~\citep{cole2017convex,cole2018approximating,caragiannis2016unreasonable,murray2020robust,kroer2019scalable,peysakhovich2019fair}. One of the most well-known algorithm for computing static market equilibria under \emph{Constant Elasticity of Substitution} (CES) utilities is the \emph{Proportional Response} (PR) dynamics \citet{zhang2011proportional,birnbaum2011distributed}.
Recently, \cite{gao2020infinite} extends the classical Fisher market model to a measurable (possibly continuous) item space and shows that infinite-dimensional EG-type convex programs capture ME under this setting.
Our work extends these ideas to a Fisher market-like scenario where items arrive sequentially.

The Fisher market literature above focuses on divisible items or randomized allocations of indivisible items. There is also a large literature on fair allocation of indivisible items (e.g. \citet{aziz2015fair,caragiannis2016unreasonable,plaut2020almost}) including approximate ME-based methods \citep{budish2011combinatorial,othman2016complexity}. 
We note that all allocations in our setting are discrete and the relationship to Fisher markets happens in the time-average sense.

Perhaps most similar to our setting is that of \cite{azar2016allocate}, who study how to allocate allocate items in an online fashion in order to obtain a market-equilibrium-like allocation. However, they consider competitive ratios, and give a primal-dual algorithm that suffers at most a logarithmic loss compared to the best hindsight optimal solution, even for worst-case arrivals. In addition to the lack of asymptotic convergence, they also only show guarantees on various (arithmetic, geometric, harmonic) averages of the utilities.
In contrast to this, our work considers stochastic arrivals, and gives an adaptive algorithm for asymptotically achieving all the desireable market equilibrium properties (e.g. no envy, Pareto optimality, equilibrium utilities). Another important difference is that our approach is easily implemented as a distributed dynamics that requires only a first-price auction allocation mechanism with indivisible allocations, in $O(n)$ time per item arrival ($n$ being the number of buyers). At each time step, the algorithm only uses buyers' valuations of the current arriving item. This makes our approach suitable for implementation in large-scale systems with a huge, possibly infinite item space.

Other methods for online fair division have been studied by various authors. 
In this literature, there are various notions of ``online:'' either buyers, items, or both can arrive online. Here we survey only related work where items arrive online.
\cite{aleksandrov2015online} studies a simple mechanism where agents can declare if they like an item, and then a coin is flipped to determine which of the agents that liked the item will get it. 
\cite{kash2014no} studies online allocation for \emph{Leontief} utilities (where each agent wants a bundle with items of fixed proportions) and shows how to achieve various properties for this setting.
\cite{cheung2019tracing} studies an evolving market environment and shows that the PR dynamics generates iterates that are close to the (changing) equilibrium. Similar to the classical PR dynamics, in each time step, all buyers' valuations of all items are known to the algorithm. 
In contrast, our work allows an \emph{unknown} underlying market from which items are sampled: in each time step, PACE only uses buyers' valuations of the current arriving item.
\cite{bateni2018fair} studies an online fair allocation problem and proposes a stochastic approximation scheme, which relies on frequently resolving the EG convex program, that ensures a constant approximation ratio (in terms of a proportional fairness measure) relative to the offline fair allocation.
Very recently, \cite{sinclair2021sequential} studies the problem of online fair division in which there is a fixed (finite) number of divisible items and sequentially arriving agents. The authors show that envy-freeness and (Pareto) efficiency cannot be minimized simultaneously; instead, there exists a boundary such that any algorithm can only possibly achieve the envy-efficiency combinations on one side of it, while they also propose such an algorithm. See also \cite{aleksandrov2020online} for a survey of further works in this area.

The idea of pacing has been studied in the context of budget management in second-price auctions. The work most related to our work is \cite{balseiro2019learning}, which studies an online version of that setting.
\cite{balseiro2019learning} and our work are very different in terms of problem settings, results, and analysis.
Here, we point out some key differences.
\cite{balseiro2019learning} first show that, in the typical setting of a single bidder interacting with second-price auctions where values and prices are drawn from a stochastic environment, an adaptive pacing strategy achieves $O(1/\sqrt{T})$ regret and ergodic convergence of the bidder's pacing multipliers. Then, under additional assumptions on monotonicity of the bidders' expected expenditures, the authors establish \emph{game-theoretic equilibrium properties} when all bidders use the same strategy, i.e., under the ``simultaneous learning'' setting. 
In contrast, we focus on \emph{market equilibrium properties} such as fairness and efficiency. Furthermore, in terms of technical assumptions and convergence results, our PACE dynamics is stepsize-free, both in theory and numerically, while the adaptive pacing algorithm in \cite{balseiro2019learning} requires careful stepsizing rules; in our setting, we have last-iterate convergence of pacing multipliers (Theorem~\ref{thm:ql-conv-beta(t)}), whereas \cite{balseiro2019learning} only establishes ergodic convergence without a rate. These differences are, fundamentally, due to the use of different first-order optimization methods, and the fact that we leverage strong convexity of the EG dual. 
The analysis in \cite{balseiro2019learning} builds upon the convergence properties of mirror descent and stochastic approximation. 
Hence, it requires pre-determined vanishing stepsizes (their Assumption 1). In the simultaneous learning setting, the stepsizes of all bidders furthermore need to be carefully selected in joint fashion (their Assumption 3). 
In contrast, PACE is completely stepsize free, thanks to structure of the reformulated dual EG convex program \eqref{eq:eg-dual-beta-bounds}, on which dual averaging can be applied directly without any stepsize parameter or an auxiliary regularization term.
This makes it vastly easier to apply PACE in practice.

\section{Proofs, Derivations and Examples} \label{app:proofs-and-deriv}

\subsection*{Proof of Theorem~\ref{eq:thm:eg-capture-me}} 
In \cite[\S 3]{gao2020infinite}, it is assumed that item supplies are uniform, i.e., $s(\theta) = 1$ for all $\theta \in \Theta$. 
As is well-known in the finite case ($\Theta = [m]$), this assumption is w.l.o.g. when studying a static Fisher market.
Here, we show that all results in \cite[\S 3]{gao2020infinite} can be easily generalized to the case of non-uniform supplies $s$ (Theorem~\ref{eq:thm:eg-capture-me}). For any market instance $M$ with buyer valuations $v_i \in L^1_+$, budgets $B_i$, $i\in [n]$ and item supplies $s\in L^\infty_+$ (all normalized as described in \ref{item:static-fm-normalization} in \S\ref{sec:fisher-markets}), consider another market instance $\tilde{M}$ with supplies being the constant function taking $1$ on $\Theta$ (denoted as $\mathbf{1}$), valuations
\[ \tilde{v}_i(\theta) = v_i(\theta) s(\theta) \]
and the same budgets. First, note that $\tilde{v}_i\in L^1_+$ since $s\in L^\infty_+$ and $\Theta$ has a finite measure ($\|s\|_\infty = \inf\{ M: |s| \leq M\, \almeve \}$):
\[ \int_\Theta \tilde{v}_i(\theta) d\theta \leq \|s\|_\infty \int_\Theta v_i(\theta)d\theta = \|s\|_\infty v_i(\Theta) = \|s\|_\infty. \]
Denote the set of feasible allocations of $M$ as $F$, that is, the set of $(x_i)$ such that $x_i \in L^\infty_+$ for all $i$ and $\sum_i x_i \leq s$. Similarly, denote the set of feasible allocations of $\tilde{M}$ as $\tilde{F}$.
For any $x\in F$, consider
\begin{align}
    \tilde{x}_i(\theta) = \begin{cases}
        x_i(\theta) / s(\theta) & {\rm if}\ \theta>0, \\
        0 & {\rm o.w.}
        \end{cases} \label{eq:xtilde-from-x}
\end{align}
Since $0\leq x_i \leq s$ (a.e.), we have $0\leq \tilde{x}_i \leq \ones$ (which means $\tilde{x}_i \in L^\infty_+$). Since $\sum_i x_i \leq s$, we have 
\[ \sum_i \tilde{x}_i \leq \ones. \]
Therefore, the set of utilities attainable by allocations in $F$ is the same as the set of utilities attainable by $\tilde{F}$. In other words, 
\[ U = \left\{ u\in \RR^n_+: x\in F,\ \langle v_i, x_i\rangle = u_i,\, i\in [n] \right\} = \tilde{U} = \left\{ \tilde{u}\in \RR^n_+: \tilde{x}\in \tilde{F}, \ \langle \tilde{v}_i, \tilde{x}_i \rangle = \tilde{u}_i \right\}. \]
By \cite[Lemma 1]{gao2020infinite} and its proof there (here, we only need the compactness, not the existence of a pure allocation for any feasible utility vector; hence, invoking \cite[Theorem 1]{dvoretzky1951relations} suffices), $\tilde{U}$ is convex and compact and so is $U$. Hence, the suprema of \eqref{eq:eg-primal} is attained.
Completely analogous to the proof of \cite[Theorem 2]{gao2020infinite}, we can show that both suprema of \eqref{eq:eg-primal} is attained. Furthermore, completely analogous to the proofs of Lemma 3 and Theorem 2 there, we can show that strong duality holds for \eqref{eq:eg-primal} and \eqref{eq:eg-dual-beta-p}. More specifically, for any $x$ feasible to \eqref{eq:eg-primal} and $(p, \beta)$ feasible to \eqref{eq:eg-dual-beta-p}, and any $0 \leq u_i \leq \langle v_i, x_i\rangle$ (i.e., adding auxiliary variable $u_i$ to \eqref{eq:eg-primal}), it holds that
\begin{align*}
    \sum_i B_i \log u_i & \leq \sum_i B_i \log u_i - \sum_i \beta_i (u_i - \langle v_i, x_i\rangle ) - \left \langle p, \sum_i x_i - s \right\rangle  \\
    & = \sum_i (B_i \log u_i - \beta_i u_i) - \sum_i \langle p - \beta_i v_i, x_i\rangle + \langle p, s\rangle\\
    & \leq \sum_i (B_i \log \frac{B_i}{\beta_i} - \beta_i \cdot \frac{B_i}{\beta_i}) - \sum_i \langle p - \beta_i v_i, x_i\rangle + \langle p, s\rangle \\
    & \leq \sum_i (B_i \log B_i - B_i) + \langle p, s\rangle - \sum_i B_i \log \beta_i \\
    & = \langle p, s\rangle - \sum_i B_i \log \beta_i - C,
\end{align*}
where the constant $C = \sum_i B_i (1-\log B_i)$. In the above derivation, the first inequality uses $\beta_i \geq 0$, $u_i \leq \langle v_i, x_i\rangle$, $\sum_i x_i \leq s$; the second inequality uses the fact that $u_i = B_i / \beta_i$ maximizes the function 
\[ u_i \mapsto B_i \log u_i - \beta_i u_i \]
for any $\beta_i > 0$ (i.e., substituting $u_i = B_i / \beta_i$ into the first line); the third inequality uses feasibility w.r.t. \eqref{eq:eg-dual-beta-p}, i.e., $p \geq \beta_i v_i$ for all $i$. 
Hence, when all inequalities are tight at a pair of solutions $x^*$ and $(p^*, \beta^*)$ feasible to \eqref{eq:eg-primal} and \eqref{eq:eg-dual-beta-p}, respectively (i.e., both optima are attained), the following KKT conditions must hold:
\begin{itemize} \setlength\itemsep{0.5em}
    \item $\langle p^*, s - \sum_i x_i^* \rangle = 0$ (via the first inequality above being tight).
    \item $u^* = B_i / \beta^*_i$ for all $i$ (via the second above).
    \item $\langle p^* - \beta^*_i v_i, x^*_i \rangle = 0$ for for all $i$ (via the third above).
\end{itemize}
As the proof of \cite[Theorem 2]{gao2020infinite} shows, these conditions (together with feasibility w.r.t. the two convex programs) are necessary and sufficient for $(x^*, p^*)$ being a ME.

\subsection*{An example: the distributional assumption on item arrivals}
In the definitions of OFM and OME in \S\ref{sec:fisher-markets}, we do not impose any distributional assumption on the sequentially arriving items $\theta_t$. 
The PACE dynamics does not require any distributional assumption either.
It is in the analysis of PACE in \S\ref{sec:conv-analysis-pace} that we assume that $\theta_t$ are drawn i.i.d. from an (unknown) underlying distribution $s \in L^\infty_+$ (where $s(\Theta)=1$ since it is a distribution). 
We define an underlying static Fisher market with the same buyers and item supplies $s$. 
Then, \S\ref{sec:conv-analysis-pace} essentially shows that PACE guarantees that various (time-averaged) quantities converge to their static equilibrium quantities in the underlying market.

To justify the necessity of such a distributional assumption on item arrivals, consider the following example in which items arrivals are chosen by an adaptive ``adversary'' whose goal is to make the buyers' time-averaged utilities of any online algorithm deviate from the static equilibrium utilities (defined by the ``hindsight market,'' i.e., the $n\times T$ market with items $\theta_1, \dots, \theta_T$) as much as possible. For simplicity (and w.l.o.g.), budgets, valuations and supplies are not normalized in this example.

\begin{itemize}\setlength\itemsep{0.5em} 
    \item There are $n$ buyers with equal budgets $B_i = 1$ and an item space of $\Theta = [n+1]$.  
    \item Let the valuation matrix be as follows, for some large $M>n$: \[ v = \begin{bmatrix}
        1      & M      & M      & M & 0 \\
        \vdots & \vdots & \vdots & 0      & M  \\
        \vdots & M      & 0      & \vdots & \vdots  \\
        1      & 0      & M          & M & M
    \end{bmatrix}. \] In other words, for item $1$, all buyers have valuation $1$. For each item $j=2, \dots, n+1$, buyer $(n+2-j)$ has valuation $0$ and other buyers have valuation $M$.
    \item The item supplies are $1$ for all $j\in [n+1]$.
    \item The number of time periods $T$ is large.
\end{itemize}


Given any fixed $j_0 =2, \dots, n+1$, the $n\times T$ static market with $T/2$ items of type $1$ and $T/2$ items of type $j_0$ exhibits the following ME:
\begin{itemize}\setlength\itemsep{0.5em}
    \item Buyer $i_0 := n+2-j^*$ receives all $T/2$ items of type $1$ with a utility of $u^*_{i_0} = T/2$ and $\beta^*_{i_0} = 2/T$. 
    \item Each buyer $i \in [n] \setminus \{i_0\}$ receives $T/(2(n-1))$ items of type $2$ (i.e., all items of type $2$ are evenly distributed among them) with a utility $u^*_i = \frac{MT}{2(n-1)}$ and $\beta^*_i = \frac{2(n-1)}{MT}$.
    \item The price of item $1$ (with $T/2$ copies) is $p^*_1 = \max \left\{ 2/T, \frac{2(n-1)}{MT} \right\} = \frac{2}{T}$. The price of item $j_0$ (with $T/2$ copies) is $p^*_{j_0} = \max\left\{ \beta^*_{i_0} \cdot 0, \frac{2(n-1)}{MT}\cdot M \right\} = 2(n-1)/T$. To verify these are equilibrium prices, note the following:
    \begin{itemize}\setlength\itemsep{0.5em}
        \item Buyer $i_0$ has $v_{i_0 j_0} = 0 < v_{i_0 1} = 1$. Hence, given $p^*$, will strictly prefer type $1$ items over type $j_0$. Her equilibrium allocation also consists of only type-$1$ items that cost exactly her budget: $\frac{T}{2}\cdot \frac{2}{T} = 1$.
        \item Buyer $i\neq i_0$ prefers type $j_0$ over type $1$ since the former has a higher value-per-unit-price: 
        \[ \frac{M}{p^*_{j_0}} = \frac{MT}{2(n-1)} > \frac{T}{2} = \frac{1}{p^*_1}. \]
        Her equilibrium allocation also consists of only type-$j_0$ items that cost exactly her budget: $\frac{T}{2(n-1)} \cdot \frac{2(n-1)}{T} = 1$. 
    \end{itemize}
\end{itemize}

Given any online algorithm, consider the following adaptive adversary.
\begin{itemize}\setlength\itemsep{0.5em}
    \item In the first $T/2$ time steps, every item arrival is of type $1$, that is, $\theta_t = 1$, $t=1, \dots, T/2$. The algorithm must irrevocably allocate them to the buyers.
    \item Since every buyer has the same valuation $1$ on type $1$, there exists a buyer $i_0$ that receives $u_{i_0} \leq T/(2n)$ up to $T/2$. 
    \item Then, the adversary picks $j_0 = n+2 - i_0$ and every item arrival in the remaining $T/2$ periods i
    s of type $j_0$ (which has value zero for $i_0$). 
\end{itemize}
In this way, buyer $i_0$ only receives a total utility of $u_{i_0}$ across all $T$. 
As shown above, in the static (hindsight) $n\times T$ market, she should have received an equilibrium utility of $u^*_{i_0} = T/2$ (i.e., being allocated all type-$1$ items). Hence, the realized utility of buyer $i_0$ is only $1/n$ of her static equilibrium utility.



\subsection*{Reformulation of \eqref{eq:eg-dual-beta-p} into \eqref{eq:eg-dual-beta-bounds}}
The reformulation is mainly based on \cite[\S 5]{gao2020infinite}, except that we now allow non-uniform supplies $s$ instead of $s(\theta)=1$ for all $\theta\in \Theta$. Assuming uniform supplies is w.l.o.g. in the static Fisher market (via rescaling all $v_i$) but is not so in OFM, since $s$ in an FOM represents the arbitrary, unknown underlying item distribution from which item arrivals are drawn.

In \eqref{eq:eg-dual-beta-p}, fixing a $\beta > 0$, setting
\[ p = \max_i \beta_i v_i \in L^1(\Theta)_+,\] 
i.e., the smallest $L^1$ function greater than or equal to $\beta_i v_i$ for all $i$, clearly minimizes the objective subject to the constraints.
Hence, we can eliminate $p$ in this way and write \eqref{eq:eg-dual-beta-p} as a finite-dimensional convex program in $\beta$. Here, $\beta_i v_i, s\rangle$ is convex in $\beta$ since $\beta \mapsto \max_i \beta_i v_i(\theta)$ is convex for any $\theta\in \Theta$. More specifically, for any $\beta, \gamma \in \RR^n_+$, $\lambda \in [0,1]$, $\theta \in \Theta$, we have
\[ \max_i\, (\lambda \beta_i + (1-\lambda) \gamma_i) v_i(\theta) \leq \lambda \max_i \beta_i v_i(\theta) + (1-\lambda)  \max_i \gamma_i v_i(\theta). \]
Hence, 
\[ \langle \max_i\, (\lambda \beta+(1-\lambda)\gamma)v_i , s\rangle \leq  \lambda \langle \max_i \beta_i v_i, s\rangle + (1-\lambda) \langle \max_i \gamma_i v_i, s\rangle. \]

Due to the strong convexity assumption in Theorem \ref{thm:general-conv}, we would need the function 
\[ \beta\mapsto \sum_i B_i \log \beta_i \] 
to be strongly convex on its domain. 
However, it is only strictly but not strongly convex on $\RR_{++}^n$. 
To resolve this, we use the following lemma. It is similar to \cite[Lemma 4]{gao2020infinite} except we allow non-uniform supplies.
\begin{lemma}
    Assume the normalizations in \ref{item:static-fm-normalization} in \S\ref{sec:fisher-markets}.
    Then, the equilibrium utilities satisfy $B_i \leq u^*_i \leq 1$ and hence $B_i \leq \beta^*_i = B_i / u^*_i \leq 1$.
    \label{lemma:bounds-on-beta}
\end{lemma}
\begin{proof}
    Since any buyer can get at most the entire set of items (given by the supply $s$), 
\[ u^*_i \leq \langle v_i, s \rangle = 1, \]
where the last inequality is due to the normalization $\langle v_i, s\rangle = 1$ in \ref{item:static-fm-normalization} in \S\ref{sec:fisher-markets}.
In any ME $(x^*, p^*)$, Theorem~\ref{eq:thm:eg-capture-me} implies
\[ \langle p^*, x^*_i \rangle = \beta^*_i  \langle v_i, x^*_i \rangle = \beta^*_i u^*_i = B_i, \] 
that is, each buyer $i$ spends her entire budget. Hence, by the normalization $\|B\|_1 = 1$ and market clearance $\langle p^*, s - \sum_ix^*_i \rangle = 0$, we have
\[ \langle p^*, s \rangle = \sum_i \langle p^*, x^*_i \rangle = \sum_i B_i = \|B\|_1 = 1 \ \Rightarrow \ \langle p^*, B_i s \rangle = B_i. \]
In other words, given item price $p^*$, each buyer $i$ can afford the proportional allocation $x^\circ_i := B_i s$. Hence, the buyer optimality property of ME implies that buyer $i$'s equilibrium utility is at least the proportional share:
\begin{align*}
    u^*_i \geq \langle v_i, x^\circ_i \rangle = B_i \langle v_i, s \rangle = B_i.
\end{align*}
Since $B_i \leq u^*_i \leq 1$ and $\beta^*_i = B_i / u^*_i$ at equilibrium (Theorem~\ref{eq:thm:eg-capture-me}), we have
\[ B_i \leq \beta^*_i \leq 1. \]

\end{proof}

By Lemma~\ref{lemma:bounds-on-beta}, adding the constraints 
\[ B_i / (1+\delta_0) \leq \beta_i \leq 1+\delta_0,\ \forall\, i\] 
to the convex program does not affect its optimal solution $\beta^*$. Here, $\delta_0 > 0$ is to ensure $\beta^*_i \in (l_i, h_i)$ (the open interval), which facilitates the convergence analysis of cumulative utilities. To simplify the constants, one can take $\delta_0 = 1$. Numerical experiments suggest that its value does not affect the speeds of convergence of quantities of interest.

Combining the above yields the reformulation \eqref{eq:eg-dual-beta-bounds}. To align \eqref{eq:eg-dual-beta-bounds} with \eqref{eq:da-stoch-opt-std-form}, for each $\theta\in \Theta$ (corresponding to $Z$ in \S\ref{sec:da-general}) and $\beta\in \RR^n_+$ (corresponding to $w$ in \S\ref{sec:da-general}), let 
\[ f_\theta(\beta) := \max_i \beta_i v_{ij}.\] 
Then, 
\[ f(\beta) := \EE f_\theta(\beta) = \langle \max_i \beta_i v_i, s\rangle,\] 
where the expectation is over $\theta\sim s$, i.e., a random variable with distribution $s$ (corresponding to $z\sim \mathcal{D}$ in \S\ref{sec:da-general}).

\subsection*{Proof of Theorem~\ref{thm:conv-beta(t)}}
It follows immediately from Theorem~\ref{thm:general-conv}, as long as the function \[ \Psi(\beta) = -\sum_i B_i \log \beta_i\] is strongly convex modulo $\sigma$ and $\EE\| v_{i_t}(\theta_t) \mathbf{e}^{(i_t)}\|^2 \leq G^2$. 
We now show them.
Note that $\Psi$ is twice differentiable and has a diagonal Hessian 
\[ 
    \nabla^2 \Psi(\beta) = \begin{bmatrix}
        \frac{B_1}{\beta_1^2} & & \\
        & \ddots & \\
        & & \frac{B_n}{\beta_n^2}
    \end{bmatrix}
\] 
at any $\beta>0$. 
Clearly, its smallest eigenvalue can be bounded as 
\[ \lambda_{\min}(\nabla^2 \Psi(\beta)) \geq \min_i \frac{B_i}{\beta_i}.\] 
Denote $\kappa = 1 / (\min_i B_i)$. 
For any $\beta$ feasible to \eqref{eq:eg-dual-beta-bounds}, by the constraints $B_i/(1+\delta_0) \leq \beta_i \leq 1+\delta_0$, we have 
\[ \lambda_{\min} (\nabla^2 \Psi(\beta)) \geq \min_i \min_{\beta_i\in [B_i/(1+\delta_0), 1+\delta_0]} \frac{B_i}{\beta_i^2} = \min_i \frac{B_i}{(1+\delta_0)^2} = \frac{1}{\kappa (1+\delta_0)^2}. \]
Therefore, $\Psi$ is strongly convex on $[B/(1+\delta_0), (1+\delta_0) \ones]$ with modulus $\sigma = \frac{1}{\kappa (1+\delta_0)^2}$.
Finally, we have
\[ \EE \| v_{i_t} \mathbf{e}^{(i_t)}\|^2 \leq \max_i \EE_{\theta\sim s} [ v_i(\theta)^2 ] = G^2 \leq \max_i \|v_i\|_\infty^2. \]

\section*{Proof of Theorem~\ref{thm:conv-utilities}}
Intuitively, our proof uses the fact that if $\beta^t_i$ and $\beta_i^*$ are near each other, then $\frac{B_i}{\beta^t_i}$ will be near $\frac{B_i}{\beta_i^*} = u_i^*$ as well. Recall that $g^t_i = u^t_i$ (i.e., the subgradient of $\beta\mapsto \max_i \beta_i v_i(\theta_t)$ that we choose corresponds to the utility buyer $i$ receives at time $t$) and hence $\bar{g}^t_i = \bar{u}^t_i$. 
Since 
\[ \beta^{t+1} = \Pi_{[l_i,h_i]}\left(\frac{B_i}{\bar{g}^t_i} \right),\] 
we know that if no projection occurs (i.e., if $\frac{B_i}{\bar{g}^t_i} \in [l_i, h_i]$) at iteration $t$, then 
\[ \frac{B_i}{\beta^{t+1}_i} = \bar{g}_i^t.\] 
Thus, we split our proof into two cases: the case where projection occurs (i.e., $\frac{B_i}{\bar{g}^t_i} \notin [l_i, h_i]$), and the case where projection does not occur. 
As we will see, the probability of a projection at time step $t$ converges to $0$ as $t$ grows.
  
    For each $i$, consider the event that no projection occurs:
    \[A^t_i := \{ l_i \leq B_i / \bar{g}^t_i \leq h_i \}. \]
    Conditioning on the complementary event $(A^t_i)^c = \{ \bar{g}^t_i \notin [l_i, h_i]\}$, it holds that 
       \[ |\beta^{t+1}_i - \beta^*_i | > \epsilon_i \ \Rightarrow \ \EE(\beta^{t+1}_i - \beta^*_i)^2 \geq \PP[(A^t_i)^c] \epsilon_i^2 \ \Rightarrow \ \PP[(A^t_i)^c]\leq \frac{1}{\epsilon_i^2}\EE(\beta^{t+1}_i - \beta^*_i)^2. \]
    Conditioning on $A^t_i$, we have $B_i / \bar{g}^t_i = \beta^{t+1}_i$. Furthermore, since 
    \[ 0\leq \bar{g}^t_i = \frac{1}{t}\sum_{\tau=1}^t v_{i j_\tau}\ind\{ i=i_\tau \} \leq \|v_i\|_\infty \]
    and $\|v_i\|_\infty \geq 1 \geq u^*_i$, we have the following upper bound on the difference between the time average of realized utilities and the equilibrium utility of buyer $i$:
    \[ |\bar{g}^t_i - u^*_i| \leq \max\{ u^*_i, \|v_i\|_\infty \} = \|v_i\|_\infty. \] 

    Now, splitting the expectation by the two complementary events $A^t_i$ and $(A^t_i)^c$, we can apply the above bounds to get
    \begin{align*}
        \EE( \bar{g}^t_i -u^*_i )^2
        &  = \EE [ \ind_{(A^t_i)^c} \cdot (\bar{g}^t_i - u^*_i)^2 ] + \EE \left[ \ind_{A^t_i} \cdot \left(\frac{B_i}{\beta^{t+1}_i} - u^*_i\right)^2\right] \\
        & \leq \|v_i\|_\infty^2\EE [ \ind_{(A^t_i)^c} ] + (u^*_i)^2 \EE \left[ \ind_{A^t_i} \cdot \left(\frac{B_i}{\beta^{t+1}_iu^*_i} - 1\right)^2\right] \\
        & \leq \|v_i\|_\infty^2 \PP[(A^t_i)^c]  + (u^*_i)^2 \cdot \EE \left( \frac{\beta^{t+1}_i - \beta^*_i}{\beta^{t+1}_i }\right)^2 \\
        & \leq \frac{\|v_i\|_\infty^2}{\epsilon_i^2} \EE (\beta^{t+1}_i-\beta^*_i)^2 + \left(\frac{(1+\delta_0) u^*_i}{B_i}\right)^2 \cdot \EE (\beta^{t+1}_i - \beta^*_i)^2 \\
        & \leq \left( \frac{\|v_i\|_\infty^2}{\epsilon_i^2} + \left( \frac{1+\delta_0}{B_i} \right)^2 \right) \EE (\beta^{t+1}_i - \beta^*_i)^2.
    \end{align*}
    Since $B_i \leq \beta^*_i \leq 1$, we have ($\kappa := 1 / (\min_i B_i)$)
    \[  \epsilon_i \geq B_i \delta_0 /(1+\delta_0) > \delta_0 / \kappa > 0. \] 
    Summing up across all $i$, using Theorem~\ref{thm:conv-beta(t)} and the above bound, we get
    \begin{align*}
      \EE \|\bar{g}^t - u^* \|^2
      & \leq \sum_i \left( \frac{\|v_i\|_\infty^2}{\epsilon_i^2} + \left( \frac{1+\delta_0}{B_i} \right)^2 \right) \EE (\beta^{t+1}_i - \beta^*_i)^2 \\
      & \leq \left(  \|v\|_\infty^2 \left(\frac{\kappa}{\delta_0}\right)^2 + ((1+\delta_0)\kappa)^2\right) \sum_i  \EE (\beta^{t+1}_i - \beta^*_i)^2 \\
        & \leq \left(  \|v\|_\infty^2 \left(\frac{\kappa}{\delta_0}\right)^2 + ((1+\delta_0)\kappa)^2\right) \frac{(6 + \log (t+1)) G^2}{(t+1) \sigma^2} \\
        & = C \cdot \frac{(6+\log (t+1)) G^2 }{(t+1) \sigma^2}.
    \end{align*}

\subsection*{Proof of Theorem~\ref{thm:conv-expenditures}}
    First, note that $\bar{b}^t_i$ can be decomposed as follows.
    \begin{align*}
        \bar{b}^t_i &= \frac{1}{t} \sum_{\tau = 1}^t \beta^\tau_i v_{i j_\tau} \ind\{i = i_\tau\} \\
        &= \beta^*_i \cdot \frac{1}{t}\sum_{\tau =1}^t v_{ij_\tau} \ind\{i = i_\tau\} + \frac{1}{t} \sum_{\tau = 1}^t (\beta^\tau_i - \beta^*_i) v_{ij_\tau} \ind\{i = i_\tau\} \\
        & = \beta^*_i \bar{g}^t_i + \frac{1}{t} \sum_{\tau = 1}^t (\beta^\tau_i - \beta^*_i) v_{ij_\tau} \ind\{i = i_\tau\}.
    \end{align*}
    Next, we bound the second term as follows, using convexity of $(\cdot)^2$ and $\|v_{i j_\tau}\| \leq \|v_i\|_\infty$:
    \[ \left(\frac{1}{t} \sum_{\tau = 1}^t (\beta^\tau_i - \beta^*_i) v_{ij_\tau} \ind\{i = i_\tau\} \right)^2 \leq \frac{1}{t} \sum_{\tau=1}^t (\beta^\tau_i - \beta^*_i)^2 \|v_i\|_\infty^2. \]
    Then, we bound the square difference between expenditure and budget as follows, using $(x+y)^2 \leq 2(x^2+y^2)$ for any $x,y\in \RR$:
    \begin{align*}
        (\bar{b}^t_i - B_i)^2 & \leq 2\left[ (\beta^*_i \bar{g}^t_i - B_i)^2 + \left(\frac{1}{t} \sum_{\tau = 1}^t (\beta^\tau_i - \beta^*_i) v_{ij_\tau} \ind\{i = i_\tau\} \right)^2\right].
    \end{align*}
    Combining the above two inequalities, taking expectation on both sides and using $\beta^*_i = B_i / u^*_i$, we have
    \begin{align}
        \EE (\bar{b}^t_i - B_i)^2 
        \leq 2\left[ (\beta^*_i)^2 \EE ( \bar{g}^t_i - u^*_i)^2 + \|v_i\|_\infty^2 \frac{1}{t} \sum_{\tau=1}^t \EE(\beta^\tau_i - \beta^*_i)^2 \right]. \label{eq:E(b_bar(t,i)-B(i))^2<=...}
    \end{align}
    When $t\geq 3$, we have $\frac{\log (t+1)}{t+1} < \frac{\log t}{t}$ (since $(\frac{\log t}{t})' = \frac{1-\log t}{t^2} < 0$ for all $t\geq 3$). By the proof of \cite[Corollary 4]{xiao2010dual}, 
    \begin{align}
        \frac{1}{t} \sum_{\tau=1}^t \frac{(6 + \log \tau) G^2}{\tau \sigma^2} \leq \frac{1}{t} \left( 6(1+\log t) + \frac{(\log t)^2}{2} \right) \frac{G^2}{\sigma^2}. \label{eq:sum-up-bounds-beta}
    \end{align}
    Finally, summing up \eqref{eq:E(b_bar(t,i)-B(i))^2<=...} across all $i$, using $\beta^*_i\leq 1$, Theorems~\ref{thm:conv-beta(t)} and~\ref{thm:conv-utilities}, and \eqref{eq:sum-up-bounds-beta}, we have
    \begin{align*}
        \EE \|\bar{b}^t - B\|^2 
        & \leq 2\left[ \EE\|\bar{g}^t - u^*\|^2 + \|v\|_\infty^2 \frac{1}{t} \sum_{\tau=1}^t \EE \|\beta^\tau - \beta^*\|^2 \right] \\
        & \leq 2\left[ C\cdot \frac{(6 + \log t) G^2}{t \sigma^2} + \|v\|_\infty^2 \frac{1}{t} \left( 6(1+\log t) + \frac{(\log t)^2}{2} \right) \frac{G^2}{\sigma^2} \right] \\
        &= \frac{2 G^2}{t\sigma^2} \left( 6(C+\|v\|_\infty^2) + (C + 6 \|v\|_\infty^2) \log t + \frac{\|v\|_\infty^2}{2} (\log t)^2 \right).
    \end{align*}

\subsection*{Proof of Theorem~\ref{thm:pace-conv-to-OME-no-regret}}

    
    For any $\theta\in \Theta$, since $\|\cdot \|_\infty$ is $1$-Lipschitz continuous w.r.t. itself, we have
    \begin{align}
        \left|p^*(\theta) - \max_i \beta^t_i v_i(\theta) \right|
        &\leq \left| \max_i \beta^*_i v_i(\theta) - \max_i \beta^t_i v_i(\theta) \right| \nonumber \\
        &\leq \max_i |\beta^*_i v_i(\theta) - \beta^t_i v_i(\theta) | \nonumber \\
        & \leq \|v\|_\infty \| \beta^t -  \beta^*\|_\infty. 
        \label{eq:bound-|p*(j)-max-beta(i,t)*v(i,j)|-any-j}
    \end{align}


    \paragraph{Analysis of regret $r^t_i$.} 
    Let $(z^\tau_i)_{\tau \in [t]} \in [0,1]^t$ be any feasible allocation on the arrived items $\theta_\tau$, $\tau \in [t]$ such that 
    \[ \frac{1}{t} \sum_{\tau = 1}^t p^\tau(\theta_t)  z^\tau_i \leq B_i. \] 
    Using $p^\tau(\theta_\tau) = \max_i \beta^\tau_i v_i(\theta_\tau)$, we have 
    \begin{align}
        \frac{1}{t} \sum_{\tau = 1}^t p^*(\theta_\tau) z^\tau_i  &= \frac{1}{t}\sum_{\tau=1}^t p^\tau(\theta_\tau) z^t_i + \frac{1}{t}\sum_{\tau=1}^t (p^*(\theta_\tau) - p^\tau(\theta_\tau)) z^\tau_i \nonumber \\
        & \leq B_i + \frac{1}{t} \|v\|_\infty \sum_{\tau = 1}^t \|\beta^\tau - \beta^*\|_\infty \quad \text{[by \eqref{eq:bound-|p*(j)-max-beta(i,t)*v(i,j)|-any-j} and $0\leq z^\tau_{ij} \leq 1$] }. 
        \label{eq:1/t*p*-vs-1/t*ptau}
    \end{align}
    Denote 
    \[\gamma_t = \frac{1}{t} \|v\|_\infty \sum_{\tau = 1}^t \|\beta^\tau - \beta^*\|_\infty. \]
    In a static ME, by Theorem~\ref{eq:thm:eg-capture-me} and the constraints in \eqref{eq:eg-dual-beta-p}, we have $p^* \geq \beta^*_i v_i$. 
    Hence,
    \begin{align}
        \frac{1}{t}\sum_{\tau = 1}^t p^*(\theta_\tau) z^\tau_i \geq \beta^*_i \left( \frac{1}{t}\sum_{\tau =1}^t v_i(\theta_\tau) z^\tau_i \right).
        \label{eq:1/t*p*-vs-1/t*v}
    \end{align}
    By \eqref{eq:1/t*p*-vs-1/t*ptau}, \eqref{eq:1/t*p*-vs-1/t*v}, $u^*_i = B_i / \beta^*_i$ (Theorem~\ref{eq:thm:eg-capture-me}) and the definition of $\xi^t_i$, we have 
    \begin{align*}
        \frac{1}{t}\sum_{\tau =1}^t v_i(\theta_\tau) z^\tau_i \leq \frac{1}{\beta^*_i}(B_i + \gamma_t)
        = u^*_i \left( 1 + \frac{\gamma_t}{B_i} \right) \leq u^*_i + \frac{\gamma_i}{B_i} \leq \bar{u}^t_i + \xi^t_i + \frac{\gamma_t}{B_i}.
    \end{align*}
    Hence, the utility level $\hat{U}^t_i$ (Definition~\ref{defn:ofm-demand-ulevel-ome}) satisfies
    \begin{align}
        \hat{U}^t_i \leq \bar{u}^t_i + \xi^t_i + \frac{\gamma_t}{B_i}.
        \label{eq:bound-hat-U-by-baru-...}
    \end{align}
    Note that $\EE(\gamma_t^2)$ can be bounded as follows:
    \begin{align}
        \EE(\gamma_t^2) \leq \|v\|_\infty^2 \frac{1}{t} \sum_{\tau=1}^t \EE\|\beta^\tau - \beta^*\|^2 \leq \frac{\|v\|_\infty^2}{t} \left( 6(1+\log t) + \frac{(\log t)^2}{2} \right)\frac{G^2}{\sigma^2} = O\left(\frac{(\log t)^2}{t}\right). \label{eq:bound-gamma(t)}
    \end{align}
    where the second inequality is due to Theorem \ref{thm:conv-beta(t)} and \eqref{eq:sum-up-bounds-beta}.
    Combining \eqref{eq:bound-hat-U-by-baru-...}, \eqref{eq:bound-gamma(t)} and $\EE(\xi^t_i)^2 = O((\log t)/t)$ (Theorem~\ref{thm:conv-utilities}), we have
    \[ \EE(r^t_i)^2 \leq 2\left( \EE(\xi^t_i)^2 + \frac{1}{B_i^2} \EE(\gamma_t^2)\right) = O\left( \frac{(\log t)^2}{t} \right). \]

    \paragraph{Analysis of envy $\rho^t_i$.} 
        Let $p^* = \max_i \beta^*_i v_i$ (a.e.) be the equilibrium prices. Similar to \ref{eq:1/t*p*-vs-1/t*ptau}, for any $i$, we have
        \begin{align}
            \frac{1}{t}\sum_{\tau =1}^t p^*(\theta_\tau) x^\tau_i 
            & = \bar{b}^t_i + \frac{1}{t}\sum_{\tau=1}^t (p^*(\theta_\tau)-\beta^\tau_i v_i(\theta_\tau)) x^\tau_i \nonumber \\
            & \leq B_i + \Delta^t_i + \eta^t_i. \label{eq:bounding-p*-x(i,t)}
        \end{align}

        Using the above (replacing $i$ with $k$) and $p^* \geq \beta^*_i v_i$, we have
        \begin{align*}
            \beta^*_i \bar{u}^t_{ik} =  \frac{1}{t}\sum_{\tau=1}^t \beta^*_i v_i(\theta_\tau) x^\tau_k 
            \leq \frac{1}{\tau} \sum_{\tau = 1}^t p^*(\theta_\tau) x^\tau_k 
            \leq B_k + \Delta^t_k + \eta^t_k.
        \end{align*}
        Hence, using $u^*_i = B_i / \beta^*_i \leq 1$ (Theorem~\ref{eq:thm:eg-capture-me} and Lemma~\ref{lemma:bounds-on-beta}),
        \begin{align*}
            \frac{\bar{u}^t_{ik}}{B_k} 
            & \leq \frac{1}{B_k} \cdot \frac{1}{\beta^*_i}(B_k + \Delta^t_k + \eta^t_k) \\
            & \leq \frac{u^*_i}{B_i} \left(1 + \frac{\Delta^t_i + \eta^t_i}{B_k} \right) \\
            & \leq \frac{\bar{u}^t_i}{B_i} + \frac{\xi^t_i}{B_i} + \frac{\Delta^t_i + \eta^t_i}{B_k B_i}\quad\text{[by definition of $\xi^t_i$]}.
        \end{align*}
        Using the above inequality, we can bound the envy as follows:
        \begin{align}
            \rho^t_i \leq \frac{\xi^t_i}{B_i} + \frac{1}{B_i} \max_k \frac{ \Delta^t_k+\eta^t_k }{B_k} \leq \kappa \xi^t_i + \kappa^2\max_k (\Delta^t_k+\eta^t_k). \label{eq:rho(t,i)<=...}
        \end{align}
        
        
        Next, we show the convergence of $\eta^t_i$. 
        By \eqref{eq:bound-|p*(j)-max-beta(i,t)*v(i,j)|-any-j}, we have
            \begin{align}
                |\eta^t_i| \leq \sum_\ell |\eta^t_\ell| \leq \frac{1}{t}\sum_{\tau=1}^t |p^*_{j_\tau} -\beta^\tau_{i_\tau} v_{i_\tau j_\tau} | \leq \frac{1}{t}\sum_{\tau=1}^t \|v\|_\infty \|\beta^\tau - \beta^*\|_\infty = \gamma_t.
                \label{eq:bound-eta(t,i)}
            \end{align}
        Hence, same as \eqref{eq:bound-gamma(t)},
            \begin{align}
                \EE (\eta^t_i)^2 \leq \|v\|_\infty^2 \frac{1}{t}\sum_{\tau=1}^t \EE \|\beta^\tau - \beta^*\|^2 \leq \|v\|_\infty^2 \frac{1}{t} \left( 6(1+\log t) + \frac{(\log t)^2}{2} \right) \frac{G^2}{\sigma^2}. \label{eq:bound-E(eta(t,i))}
            \end{align}
            By Theorems~\ref{thm:conv-beta(t)} and~\ref{thm:conv-expenditures}, we know that $\EE (\xi^t_i)^2 = O\left((\log t) / t \right)$ and $\EE(\Delta^t_i)^2 = O\left((\log t)^2 / t\right)$. 
            Together with \eqref{eq:bound-E(eta(t,i))} and \eqref{eq:rho(t,i)<=...}, we have
            \[ \EE(\rho^t_i)^2 \leq \kappa \EE(\xi^t_i)^2 + \kappa^2 \sum_\ell (\EE (\Delta^t_\ell)^2 + \EE(\eta^t_\ell)^2) = O\left(\frac{(\log t)^2}{t}\right). \]

\section{Extension to quasilinear utilities} \label{app:ql}
We show that PACE can be easily extended to the case of a \emph{quasilinear} (QL) market (i.e., where buyers have QL utilities). We show that most of the convergence results in \S\ref{sec:conv-analysis-pace} still holds. 
The static quasilinear market setup is the same as the linear case in \S\ref{sec:fisher-markets} (which allows a possibly infinite item space $\Theta$), except the following: 
\begin{itemize} \setlength\itemsep{0.5em}
    \item For given item prices $p\in L^1_+$, 
    each buyer $i$ has a quasilinear utility function, i.e., 
    \[ u_i(x_i) = \langle v_i, x_i \rangle - \langle p, x_i\rangle. \]
    \item Without loss of generality, assume $\|B\|_1 = 1$ and all buyers' valuations are nontrivial, i.e., $\langle v_i, s\rangle >0$ for all $i$. 
    Due to the structure of QL utilities, we cannot normalize the valuations $v_i$ and budgets $B_i$ separately without loss of generality. Instead, they can only be scaled at the same time by the same constant.
\end{itemize}
Same as before, each buyer $i$ has a budget $B_i>0$ and can only choose among budget-feasible allocations, that is, $x_i$ such that $\langle p, x_i \rangle \leq B_i$. In this case, an allocation-price pair $(x^*, p^*)$ is a \emph{quasilinear market equilibrium} (QLME) if the following holds (see \cite[\S 6]{gao2020infinite} and \cite[\S 4]{cole2017convex}):
\begin{itemize} \setlength\itemsep{0.5em}
    \item Buyers are optimal: $x^*_i \in D_i(p^*) := \argmax \{ \langle v_i - p^*, x_i \rangle: x_i \in L^\infty_+, \, \langle p^*, x_i \rangle \leq B_i\}$.
    \item The market clears: $\sum_i x_i \leq s$ and $\langle p^*, s - \sum_i x^*_i \rangle = 0$.
\end{itemize}

As shown in \cite[\S 6]{gao2020infinite}, the following pair of (possibly infinite-dimensional) convex programs capture QLME:\footnote{There, the authors assume $s = \ones$, which is w.l.o.g. for static Fisher markets. 
Similar to the case of Theorem~\ref{eq:thm:eg-capture-me} for linear utilities, all results can be easily extended to the case of $s\in L^\infty_+$.}
\begin{align}
    \begin{split}
    \sup\, & \sum_i (B_i \log u_i - \delta_i) \\ 
    {\rm s.t.} & u_i \leq \langle v_i, x_i \rangle + \delta_i,\, \forall\, i \in [n], \\
    & \sum_i x_i \leq s,\,  \\
    & u_i \geq 0,\ \delta_i \geq 0,\ x_i \in L_1(\Theta)_+,\ \forall\, i \in [n].
    \end{split}
    \label{eq:ql-eg-primal} 
    \tag{$\mathcal P_{\rm QLEG}$}
\end{align}
\begin{align}
    \begin{split}
        \inf\, & \langle p, s \rangle - \sum_i B_i \log \beta_i \\
        {\rm s.t.} & p \geq \beta_i v_i,\ \beta_i \leq 1,\ \forall\, i \in [n], \\
        & p \in L_1(\Theta)_+,\ \beta\in \RR^d_+.
    \end{split} 
    \label{eq:ql-eg-dual-p-beta}
    \tag{$\mathcal D_{\rm QLEG}$}
\end{align}

In the sequel, we use $(x^*, u^*, \delta^*)$ to denote an optimal solution of \eqref{eq:ql-eg-primal} (in which $u^*$ and $\delta^*$ are unique) and $(p^*, \beta^*)$ to denote the optimal solution of \eqref{eq:ql-eg-dual-p-beta}. As shown in \cite[\S 6]{gao2020infinite}, 
the following KKT conditions of \eqref{eq:ql-eg-primal} and \eqref{eq:ql-eg-dual-p-beta} are necessary and sufficient for $(x^*, p^*)$ being a QLME.
\begin{itemize} \setlength\itemsep{0.5em}
    \item $\delta^*_i (1 - \beta^*_i) = 0$ for all $i$ (complementary slackness).
    \item $u^*_i = B_i / \beta^*_i$ for all $i$.
    \item $p^* = \max_i \beta^* v_i$ (a.e.) for all $j$.
    \item $\langle p^* - \beta^*_i v_i, x^*_i \rangle = 0$ for all $i$.
\end{itemize}
Let $(x^*, p^*)$ denote a QLME. The equilibrium utility of buyer $i$ (i.e., the amount of utility buyer $i$ receives at a QLME) is 
\[ u^{\rm QLME}_i := \langle v_i - p^*, x^*_i \rangle = (1-\beta^*_i) \langle v_i, x^*_i \rangle = (1-\beta^*_i)(u^*_i - \delta^*_i), \]
which is unique and does not depend on the choice of the equilibrium allocation $x^*$. 
In general, $u^{\rm QLME}_i$ is not the same as $u^*_i$ in the optimal solution of \eqref{eq:ql-eg-primal}. 
In comparison, the term $\langle v_i, x^*_i\rangle$ can be viewed as the equilibrium gross utility before subtracting the price $\langle p^*, x^*_i \rangle$ of the allocation $x^*_i$. 
The above equilibrium quantities satisfy the following \citet[\S 6]{gao2020infinite}.
\begin{itemize} \setlength\itemsep{0.5em}
    \item If $\beta^*_i = 1$, then $\langle p^* - \beta^*_i v_i, x^*_i \rangle = 0$ implies its gross utility and expenditure are equal, which give an equilibrium utility of zero:
        \[ \langle v_i, x^*_i \rangle = \beta^*_i\langle v_i, x^*_i \rangle = \langle p^*, x^*_i \rangle = u^*_i - \delta^*_i \ \Rightarrow \ u^{\rm QLME}_i = \langle v_i - p^*, x^*_i\rangle = 0. \]
    \item If $\beta^*_i < 1$, then $\delta^*_i = 0$ by complementary slackness (the first KKT condition above). Hence, the gross utility is $\langle v_i, x^*_i \rangle = u^*_i$ and
    \[ u^{\rm QLME} = \langle v_i - p^*, x^*_i \rangle = (1-\beta^*_i)\langle v_i, x^*_i \rangle = (1-\beta^*_i) u^*_i. \]
\end{itemize}
Similar to the proof of \cite[Lemma 5]{gao2020first}, we can show that
\[ u^*_i \leq \langle v_i, s \rangle + B_i. \]
Hence,
\[ \beta^* = \frac{B_i}{u^*_i} \geq \frac{B_i}{\langle v_i, s \rangle + B_i} > \beta^{\min}_i := \frac{B_i}{\langle v_i, s \rangle + 2 B_i} > 0. \]
The choice of $\beta^{\min}_i$ is to ensure that $\beta^*_i - \beta^{\min}_i > 0$, which simplifies the analysis of the dynamics. 
Substituting $p = \max_i \beta_i v_i$ and using the bounds $ \beta^{\min}_i \leq \beta^*_i \leq 1$, we can solve the following convex program for the equilibrium utility prices $\beta^*$, where $\beta^{\min} := (\beta^{\min}_1, \dots, \beta^{\min}_n)$:
\begin{align}
    \min_{ \beta\in [\beta^{\min}, \ones]} \ \langle p, s\rangle - \sum_i B_i \log \beta_i.  \label{eq:ql-dual-beta-cp-bounds}
\end{align}
Applying dual averaging to the convex program \eqref{eq:ql-dual-beta-cp-bounds}, we arrive at the following PACE dynamics for an online QL market (i.e., an OFM with buyers having QL utilities). 
At time $t$, the following steps take place.
\begin{itemize} \setlength\itemsep{0.5em}
    \item An item $\theta_t \in \Theta $ arrives, which determines a winner $i_t = \min \argmax_i \beta^t_i v_i(\theta_t)$. 
    \item The stochastic subgradient is $g^t = v_{i_t}(\theta_t) \mathbf{e}^{(i_t)}$, or $g^t_i = v_i(\theta_t) \ind\{ i=i_t \}$ for each $i$.
    \item Each buyer $i$ pays a price (expenditure)  
        \[ b^t_i =  \beta^t_i v_i(\theta_t) \ind\{ i=i_t \} \] 
        and receives a (net) utility of
    \begin{align}
        u^t_i = g^t_i - b^t_i = (1 - \beta^t_i) v_i(\theta_t) \ind\{ i=i_t\},
        \label{eq:u(t,i)=g(t,i)-p(t,i)-ql}
    \end{align}
    which is is the value of the item minus the price paid.
    Here, only the winning buyer $i_t$ may get a potentially nonzero utility $u^t_{i_t}$; other buyers $i\neq i_t$ gets $0$ (and pays zero).
    \item Update the dual average: for each $i$, $\bar{g}^t = \frac{t-1}{t}\bar{g}^{t-1} + \frac{1}{t}g^t$, which ensures $\bar{g}^t_i = \frac{1}{t}\sum_{\tau=1}^t v_{ij_\tau}\ind\{i=i_\tau\}$ for all $i$ (same as in the linear case).
    \item Compute the next pacing multiplier (similar to the linear case, except the lower bound for $\beta^t_i$ being $\beta^{\min}_i$ instead of $B_i$):
        \[ \beta^{t+1}_i = \argmin_{\beta_i \in [\beta^{\min}_i, 1]} \left\{\bar{g}^t_i \beta_i - B_i \log \beta_i \right\}  \ \Rightarrow\ \beta^{t+1}_i = \Pi_{ [\beta^{\min}_i, 1] } \left(\frac{B_i}{\bar{g}^t_i}\right). \]
\end{itemize}
Same as in the linear case, we do not need any distributional assumption on the item arrivals to run PACE. In subsequent convergence analysis, however, we assume that the items $\theta_t$ are drawn i.i.d. from a distribution $s$ (i.e., $s\in L^\infty_+$ and $s(\Theta)=1$). We also assume that $v_i \in L^\infty_+$ for all $i$.
Let $(x^*, p^*)$ denote a QLME of the underlying static QL market with supplies $s$.
\paragraph{Convergence of QL pacing multipliers.}
Analogous to Theorem~\ref{thm:conv-beta(t)}, in the QL case, we can show that the pacing multipliers $\beta^t$ converge to the equilibrium utility prices $\beta^*$ in mean square. It is a direct consequence of the general convergence result of dual averaging (Theorem~\ref{thm:general-conv}).
\begin{theorem}
    For $t=1,2,\dots$, it holds that 
    \[ \EE\|\beta^t - \beta^* \|^2 \leq \frac{(6+\log t) G^2}{t\sigma^2}, \]
    where $G^2$, $\kappa$ are the same as in Theorem \ref{thm:conv-beta(t)} and 
    \[ \sigma = \min_i \min_{\beta_i \in [\beta^{\min}_i, 1]} \frac{B_i}{\beta_i^2} = \min_i B_i = \frac{1}{\kappa}. \]
    \label{thm:ql-conv-beta(t)}
\end{theorem}

\paragraph{Convergence of QL utilities and expenditures.}
Next, we show mean-square convergence of time-averaged utilities 
$\bar{u}^t = \frac{1}{t}\sum_{\tau=1}^t u^t_i$ and expenditures $\bar{b}^t = \frac{1}{t}\sum_{\tau=1}^t \beta^\tau_i v_{ij_\tau}\ind\{i = i_\tau\}$.
In the QL case, the dual average $g^t_i$ can also be viewed as the per-period gross utility before subtracting the price $\beta^t_i v_i(\theta_t)$. 
In this way, $\bar{g}^t_i$ is the time-averaged gross utility of buyer $i$.
\begin{theorem}
    For $t = 1, 2, \dots$ and each $i$, the following holds.
    \begin{itemize}
        \item If $\beta^*_i = 1$, then $u^{\rm QLME}_i = 0$ and $B_i = u^*_i$. In this case, 
        \begin{align}
            & \EE \left(u^t_i - u^{\rm QLME}_i\right)^2 = \EE |u^t_i|^2 \leq \|v_i\|_\infty^2 \EE (1 - \beta^t_i)^2  = O\left( \frac{\log t}{t} \right). 
            \label{eq:ql-util-conv-beta=1-case}
        \end{align}
        \item If $\beta^*_i < 1$, then $u^{\rm QLME}_i > 0$, $\delta^*_i = 0$ and $u^{\rm QLME} = (1-\beta^*_i) u^*_i$. 
        In this case, the gross utility $\bar{g}^t_i$, realized (net) utility $\bar{u}^t_i$ and expenditures $\bar{b}^t_i$ converge as follows:
            \begin{align*}
                & \EE( \bar{g}^t_i -u^*_i )^2 \leq C_i \EE (\beta^{t+1}_i - \beta^*_i)^2 = O\left( \frac{\log t}{t} \right),\\
                & \EE \left(\bar{u}^t_i - u^{\rm QLME}_i \right)^2 \leq R^t_i, \ \ 
                 \EE (\bar{b}^t_i - B_i)^2 \leq R^t_i,
            \end{align*}
            where 
            \begin{align*}
                C_i &= \frac{\|v_i\|_\infty^2}{\epsilon_i^2} + \frac{\left( \frac{\|v_i\|_1}{m} + 2B_i \right)^4}{B_i^2}, \ \ \epsilon_i = \min\{ 1 - \beta^*_i, \beta^*_i - \beta^{\min}_i \}, \\ 
                R^t_i &= 2 \left[ \EE (\bar{g}^t_i - u^*_i)^2 + \frac{\|v_i\|_\infty^2 }{t} \sum_{\tau=1}^t \EE (\beta^\tau_i -\beta^*_i)^2 \right], \ \ \EE R^t_i = O\left(\frac{(\log t)^2}{t}\right).
            \end{align*}
    \end{itemize}
    Hence, the mean-square error $\EE \left\| \bar{u}^t - u^{\rm QLME} \right \|^2$ is either $O((\log t)^2/t)$ (when some $\beta^*_i<1$) or $O((\log t)/t)$ (when all $\beta^*_i = 1$).
    \label{thm:ql-conv-g(t)}
\end{theorem}
\begin{proof}
The $\beta^*_i = 1$ case is clear. We prove the $\beta^*_i < 1$ case. 

\paragraph{Convergence of $\bar{g}^t_i$.}
Denote the event $A^t_i = \{ B_i \leq \bar{g}^t_i \leq B_i / \beta^{\min}_i \}$.
Then, $(A^t_i)^c$ means $ B_i / \bar{g}^t_i \notin [ \beta^{\min}_i, 1 ]$ and hence $|\beta^{t+1} - \beta^*_i | > \epsilon_i$.
Similar to the linear case, we deduce
\[ \EE(\beta^{t+1}_i - \beta^*_i)^2 \geq \PP[(A^t_i)^c] \epsilon_i^2 \Rightarrow \ \PP[(A^t_i)^c]\leq \frac{1}{\epsilon_i^2}\EE(\beta^{t+1}_i - \beta^*_i)^2. \]
Furthermore, since $ 0 \leq \bar{g}^t_i \leq \|v_i\|_\infty$ (same as in the linear case) and $u^*_i \leq \langle v_i, s\rangle + B_i$, we have
\begin{align*}
    \EE( \bar{g}^t_i -u^*_i )^2
    &  = \EE [ \ind_{(A^t_i)^c} \cdot (\bar{g}^t_i - u^*_i)^2 ] + \EE \left[ \ind_{A^t_i} \cdot \left(\frac{B_i}{\beta^{t+1}_i} - u^*_i\right)^2\right] \\
    & \leq \|v_i\|_\infty^2\EE [ \ind_{(A^t_i)^c} ] + (u^*_i)^2 \EE \left[ \ind_{A^t_i} \cdot \left(\frac{B_i}{\beta^{t+1}_iu^*_i} - 1\right)^2\right] \\
    & \leq \|v_i\|_\infty^2 \PP[(A^t_i)^c]  + (u^*_i)^2 \cdot \EE \left( \frac{\beta^{t+1}_i - \beta^*_i}{\beta^{t+1}_i }\right)^2 \\
    & \leq \frac{\|v_i\|_\infty^2}{\epsilon_i^2} \EE (\beta^{t+1}_i-\beta^*_i)^2 + \left(\frac{u^*_i}{\beta^{\min}_i}\right)^2 \cdot \EE (\beta^{t+1}_i - \beta^*_i)^2 \\
    & \leq \left( \frac{\|v_i\|_\infty^2}{\epsilon_i^2} + \left( \frac{\left( \langle v_i, s\rangle +B_i\right) \left( \langle v_i, s\rangle + 2B_i\right) }{B_i} \right)^2 \right) \EE (\beta^{t+1}_i - \beta^*_i)^2 \\
    & = C_i \EE (\beta^{t+1}_i - \beta^*_i)^2 = O\left( \frac{\log t}{t} \right).
\end{align*}

\paragraph{Convergence of expenditures $\bar{b}^t_i$.} 
Similar to the linear case, note that $\bar{b}^t_i$ can be decomposed as follows (where $x^\tau_i := \ind\{ i=i_\tau\}$ denotes whether buyer $i$ wins at time step $\tau$):
\begin{align*}
    \bar{b}^t_i := \frac{1}{t}\sum_{\tau = 1}^t \beta^\tau_i v_i(\theta_\tau)x^\tau_i = \beta^*_i \bar{g}^t_i + \frac{1}{t}\sum_{\tau=1}^t (\beta^\tau_i-\beta^*_i)v_i(\theta_\tau)x^\tau_i.
\end{align*}

Hence, using $\beta^*_i = B_i / u^*_i \leq 1$, $(x+y)^2 \leq 2(x^2+y^2)$, convexity of $(\cdot)^2$ and $v_i(\theta_t)x^\tau_i \leq \|v_i\|_\infty$, we have
\begin{align}
    \EE (\bar{b}^t_i - B_i)^2 & \leq 2 \left[  \EE (\beta^*_i \bar{g}^t_i - B_i)^2 +  \EE \left( \frac{1}{t}\sum_{\tau=1}^t (\beta^\tau_i-\beta^*_i)v_i(\theta_\tau) x^\tau_i \right)^2\right] \nonumber \\
    & \leq 2 \left[ (\beta^*_i)^2 \EE (\bar{g}^t_i - u^*_i)^2 + \frac{\|v_i\|_\infty^2 }{t} \sum_{\tau=1}^t \EE (\beta^\tau_i -\beta^*_i)^2  \right] \nonumber \\
    & \leq 2 \left[ \EE (\bar{g}^t_i - u^*_i)^2 + \frac{\|v_i\|_\infty^2 }{t} \sum_{\tau=1}^t \EE (\beta^\tau_i -\beta^*_i)^2  \right] = R^t_i.
    \label{eq:ql-bounding-|b(t,i)-B(i)|}
\end{align}
The order of $\EE R^t_i$ is given by those of $\EE (\bar{g}^t_i - u^*_i)^2$ and $\sum_{\tau = 1}^t \EE (\beta^\tau_i - \beta^*_i)$, which are $O((\log t)/t)$ and  $O((\log t)^2 / t )$, respectively.

\paragraph{Convergence of utilities $\bar{u}^t_i$.}
For a buyer $i$ with $\beta^*_i < 1$, let 
 \[ 
    \epsilon_i = \min\{ 1 - \beta^*_i, \beta^*_i - \beta^{\min}_i \} > 0. \footnote{The analysis of this case also works for $\beta^*_i = 0$ but its resulting bound is not as tight as the above one for the case $\beta^*_i = 1$.}
 \] 
 Express $\bar{u}^t_i$ as follows:
\begin{align*}
    \bar{u}^t_i &= \frac{1}{t} \sum_{\tau=1}^t (1- \beta^\tau_i) v_i(\theta_\tau)x^\tau_i \\ 
    & = (1 - \beta^*_i) \bar{g}^t_i + \frac{1}{t} \sum_{\tau=1}^t (\beta^\tau_i - \beta^*_i) v_i(\theta_\tau) x^\tau_i. 
\end{align*}
Since $u^{\rm QLME}_i = (1 - \beta^*_i) u^*_i$, similar to \eqref{eq:ql-bounding-|b(t,i)-B(i)|}, we have
\begin{align*}
    \EE \left(\bar{u}^t_i - u^{\rm QLME}_i\right)^2 
    & \leq 2 \left[ (1-\beta^*_i) \EE (\bar{g}^t_i - u^*_i)^2 + \EE \left( \frac{1}{t} \sum_{\tau=1}^t (\beta^\tau_i - \beta^*_i) v_i(\theta_\tau) x^\tau_i  \right)^2 \right] \\
    & \leq 2\left[ \EE (\bar{g}^t_i - u^*_i)^2 + \frac{\|v_i\|_\infty^2 }{t} \sum_{\tau=1}^t \EE (\beta^\tau_i -\beta^*_i)^2  \right] = R^t_i. 
\end{align*}


Finally, if all $\beta^*_i =1$, \eqref{eq:ql-util-conv-beta=1-case} implies that $\EE \|\bar{u}^t - \bar{u}^{\rm QLME}\|^2 = O((\log t)/t)$. It some $\beta^*_i < 1$, since $\EE R^t_i = O((\log t)^2/t)$, so is $\EE \|\bar{u}^t - \bar{u}^{\rm QLME}\|^2$. 
\end{proof}

\section{More details on the experiments} \label{app:experiment-details}
In each experiment, we will have some underlying valuations, items will be drawn one-at-a-time, uniformly at random, from the set of possible items, on which we run the PACE dynamics. 
We have several outcome measures of interest for asking how close we are to the static equilibrium quantities at each point.
First, we look at convergence of realized utilities. 
In each case we consider the realized utilities up to time $t$ and look at the deviation from equilibrium utility normalized by the equilibrium utility level. We look at both the average and the worst-case deviations. Formally these are calculated as $\|(\bar{u}^t - u^*) / u^*\|_1 / n$ for the average deviation and $\|(\bar{u}^t - u^*)/ u^*\|_\infty$ for the maximum (over buyers) deviation. 
We also measure deviations of the pacing multiplier $\beta^t$ from  $\beta^*$ 
and deviations of time-averaged cumulative expenditure $\bar{b}^t$ from buyers' budgets $B = (B_1, \dots, B_n)$ 
using analogous normalizations. In the plots, we add horizontal lines for the same error measures for the proportional shares of the static underlying Fisher market (each buyer receiving $B_i$ of each item), a `baseline' solution.

We consider 3 different market datasets. The first two datasets are recommender systems which we turn into markets. The final is taken from a survey experiment. We point the reader to \cite{kroer2019computing} for a more in-depth discussion and exploratory data analysis of these $3$ datasets. The first dataset uses MovieLens \citep{harper2015movielens}. MovieLens is a dataset of individual ratings of movies, \citep{kroer2019computing} turn it into a market by using matrix completion to fill in missing user-movie ratings, they then take the top 1500 most active users and 1500 most rated movies and set the valuations $v_{ij}$ as the predicted ratings from the matrix completion. We also use the Jester Jokes dataset \citep{goldberg2001eigentaste}. 
Here, we have $7200$ individuals that have rated $100$ jokes. We treat the jokes as the item to be allocated.Finally, we use the Household Items dataset introduced in \citep{kroer2019computing}. Here we have $2876$ survey takes entering a willingness to pay for $50$ household items (vacuum cleaners, toasters, gas grills, etc.). 
For each dataset, we first rescale (w.l.o.g.) buyer valuations as described in \S \ref{sec:conv-analysis-pace}.

We also consider an experiment on a simple infinite-dimensional market instance (which we refer to as ``Inf-Dim'') of $n=100$ buyers and item space $\Theta = [0,1]$, similar to the examples in \cite[\S 4.2]{gao2020infinite}. Let each buyer valuation $v_i$ be normalized linear functions on $[0,1]$, that is, $v_i(\theta) = c_i(\theta) + d_i$ such that $v_i(\Theta) : = \int_\Theta v_i d\mu = \int_0^1 v_i(\theta) d\theta = 1 \ \Leftrightarrow \ \frac{c_i}{2} + d_i = 1$. We randomly generate $(c_i, d_i)$, $i=1, \dots, n$ and run the dynamics for $T = 100 n$ time steps.

For the finite dimensional datasets we compute equilibrium utilities $u^*$ and utility prices $\beta^*$ by solving the corresponding static instances using standard methods. 
For the infinite dimensional synthetic instance, we use the approach based on convex conic reformulation \cite[\S 4]{gao2020infinite} to compute $\beta^*$. 

Figure~\ref{fig:plot-3-datasets} in \S\ref{sec:experiments} contains the plots for the MovieLens, Household Items and Inf-Dim datasets. Figure~\ref{fig:jokes-app} contains the plots for the Jokes dataset.

Since items arrive one at a time, $t=100$ time steps in a market with $n=10$ buyers is very different from the same number of time steps in a market with $n=1000$ buyers. To deal with this, we run PACE for $T = 100 n$ time steps, referring to each $n$ time steps as an \emph{epoch}.

We record the average and maximum values of relative errors of the pacing multipliers $\beta^t$, time-averaged cumulative utilities $\bar{u}^t$ and time-averaged expenditures $\bar{b}^t$. 

\paragraph{Convergence of expenditures to total budget.}
For each $i$, the quantity 
\[ \left|\frac{\bar{b}^t_i - B_i}{B_i} \right| = \left| \frac{\sum_{\tau=1}^t b^t_i - tB_i }{ t B_i} \right|\] 
can be viewed as the relative deviation of current cumulative expenditure at time $t$ from the total budget $tB_i$ available up to $t$. Hence, the residuals $\left\| (\bar{b}^t_i - B) / B \right\| / n$ and $\left\| (\bar{b}^t_i - B) / B \right\|_\infty$ are the average and maximum such deviations across all buyers. 
For each dataset (MovieLens, Household, Jokes and Inf-Dim), we plot the various quartiles of these residuals across all seeds, as shown in Figure~\ref{fig:quartiles-spending-vs-budget}.
\begin{center}
    \begin{figure}
        \centering
        \includegraphics[scale=0.33]{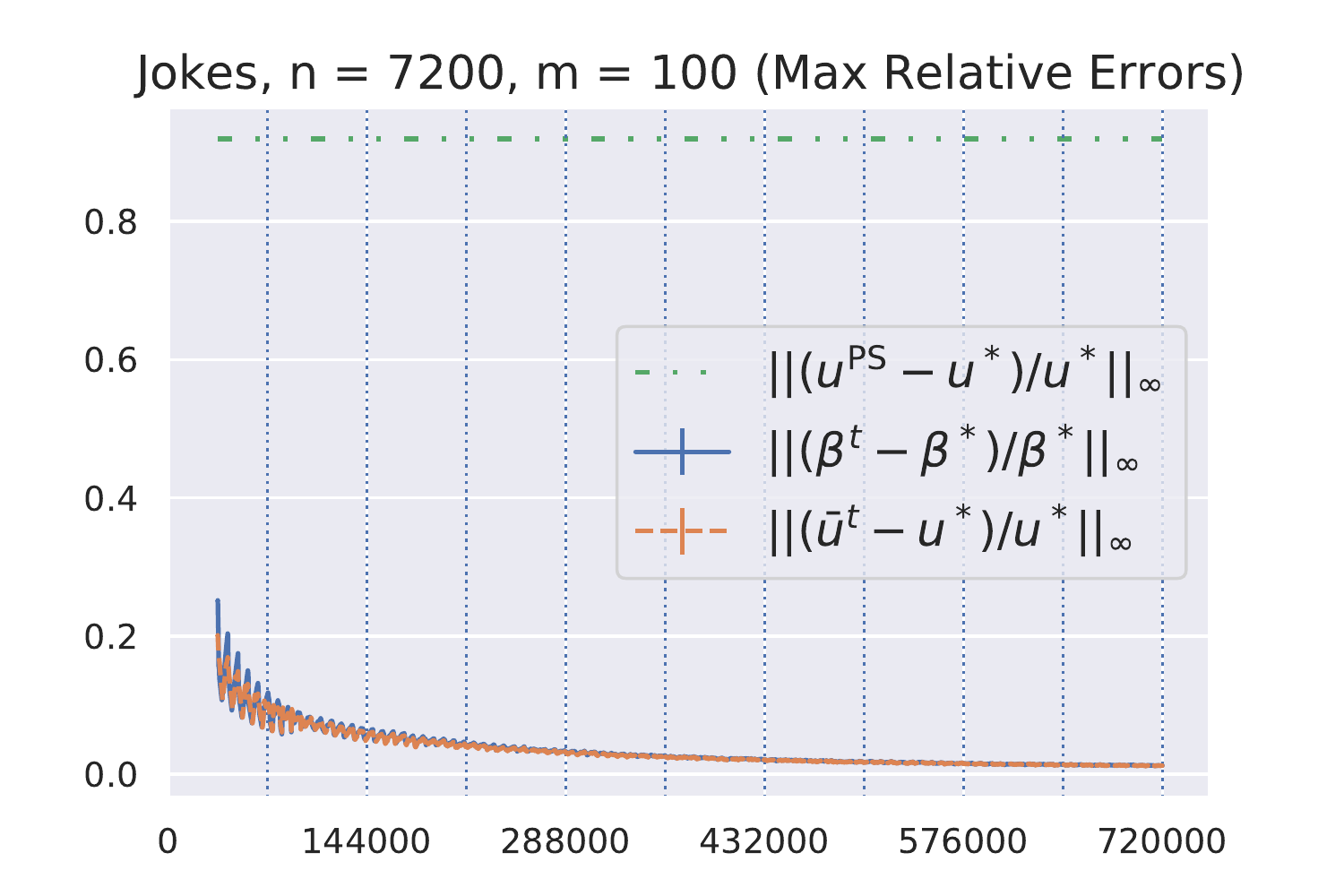} 
        \includegraphics[scale=0.33]{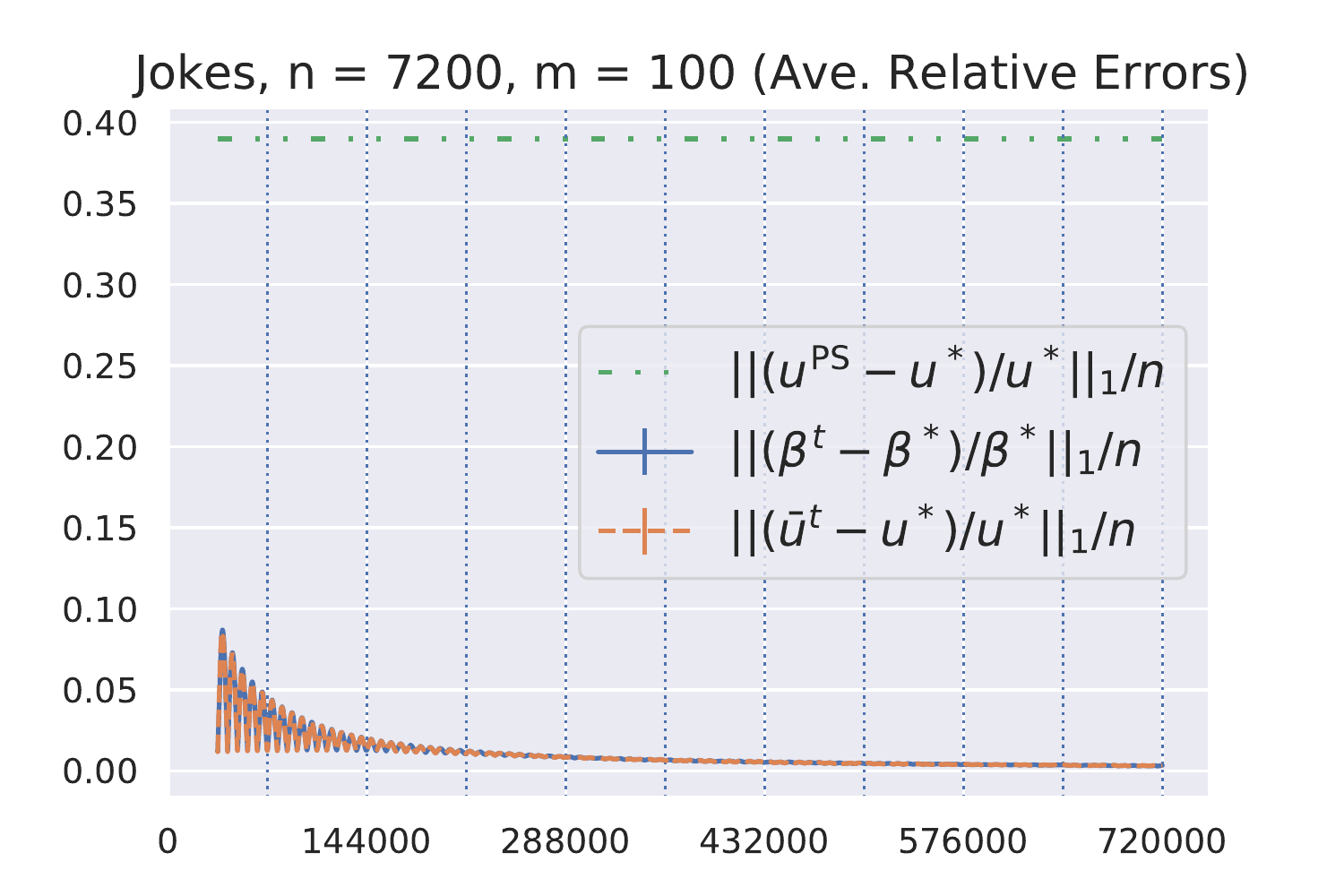}
        \caption{Results for the same experiments as in Figure~\ref{fig:plot-3-datasets} (convergence of pacing multipliers, utilities and expenditures) on the Jokes dataset.}
        \label{fig:jokes-app}
    \end{figure}
\end{center}

\begin{center}
    \begin{figure}
        \centering
        \includegraphics[scale=0.35]{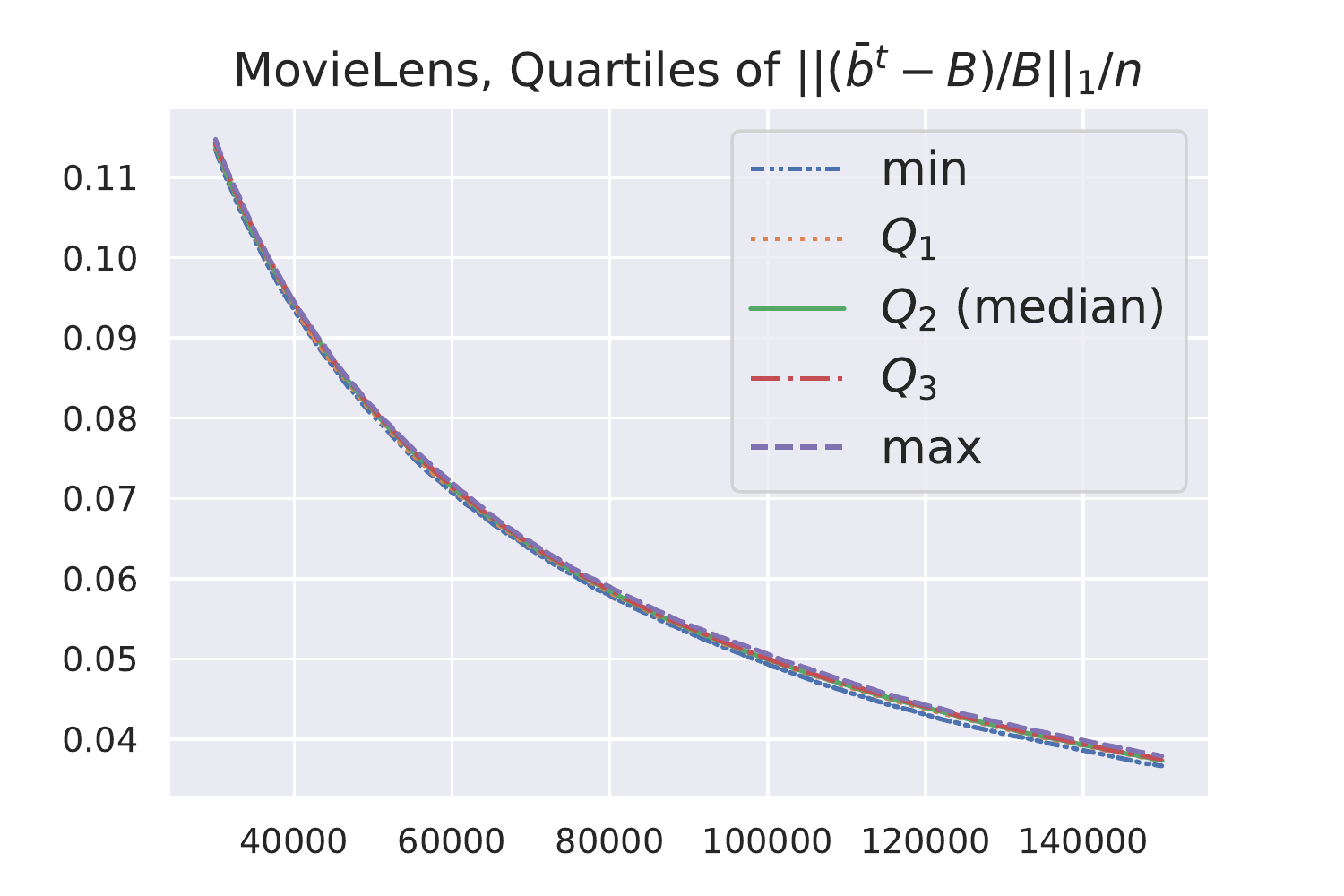} 
        \includegraphics[scale=0.35]{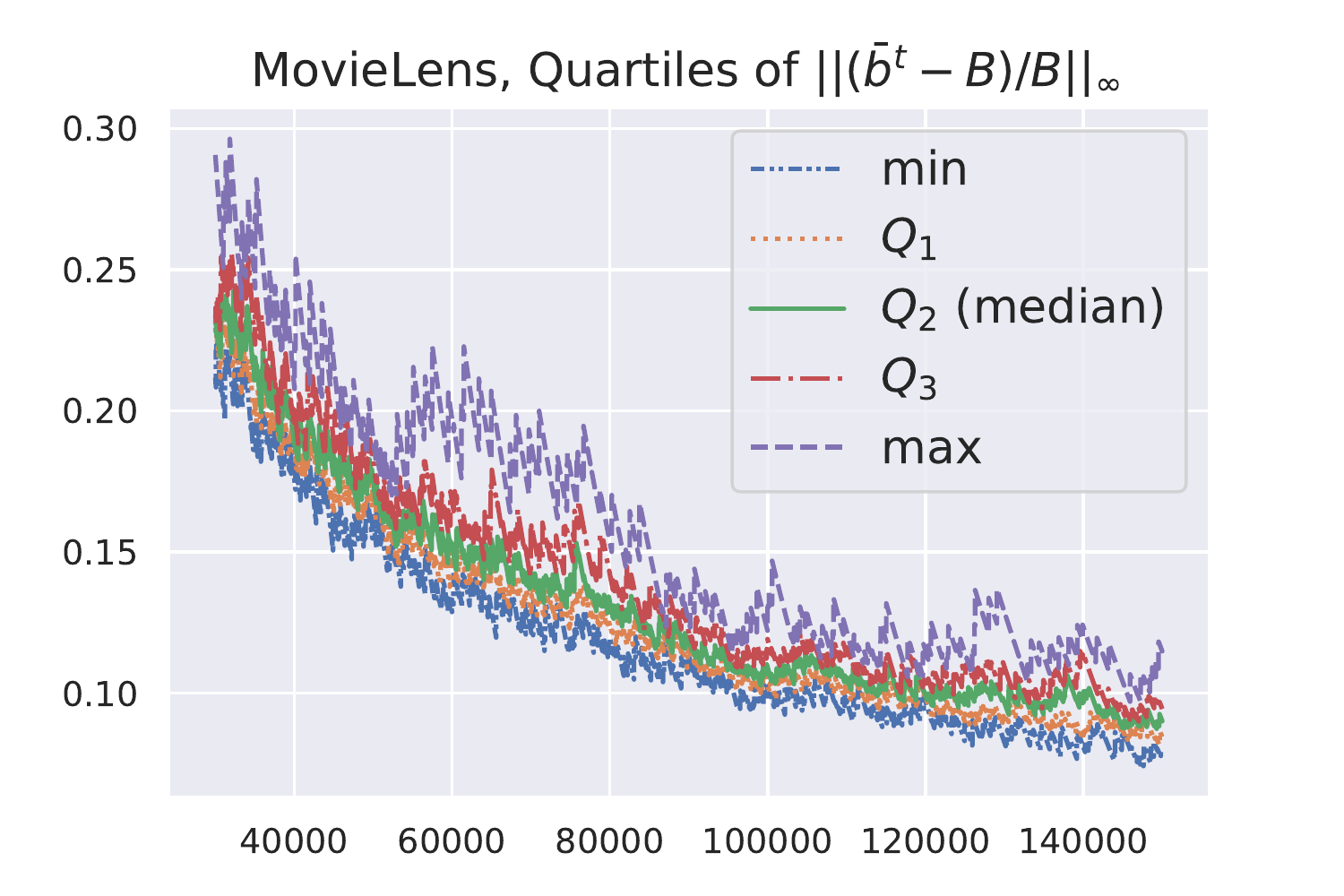}

        \includegraphics[scale=0.35]{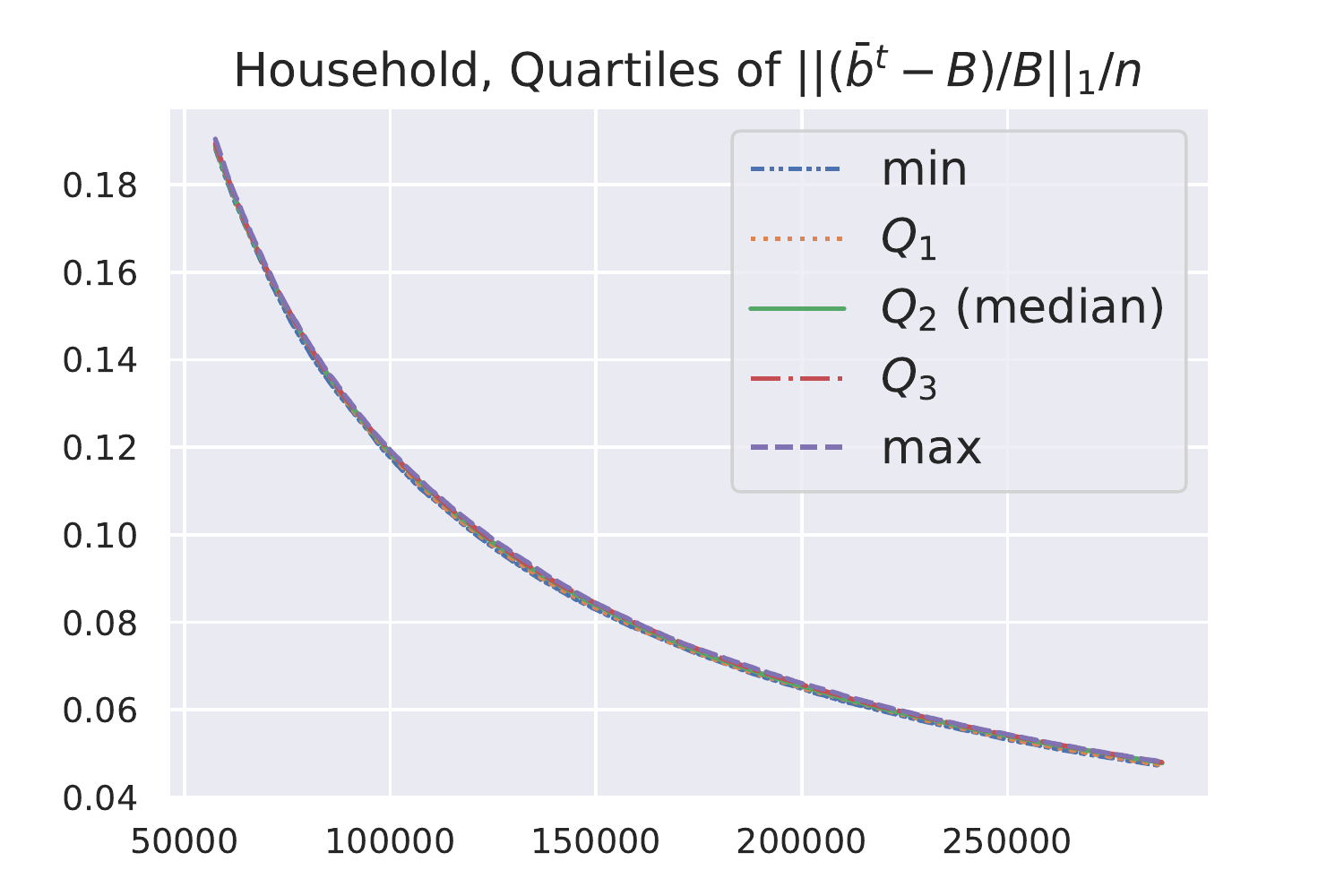} 
        \includegraphics[scale=0.35]{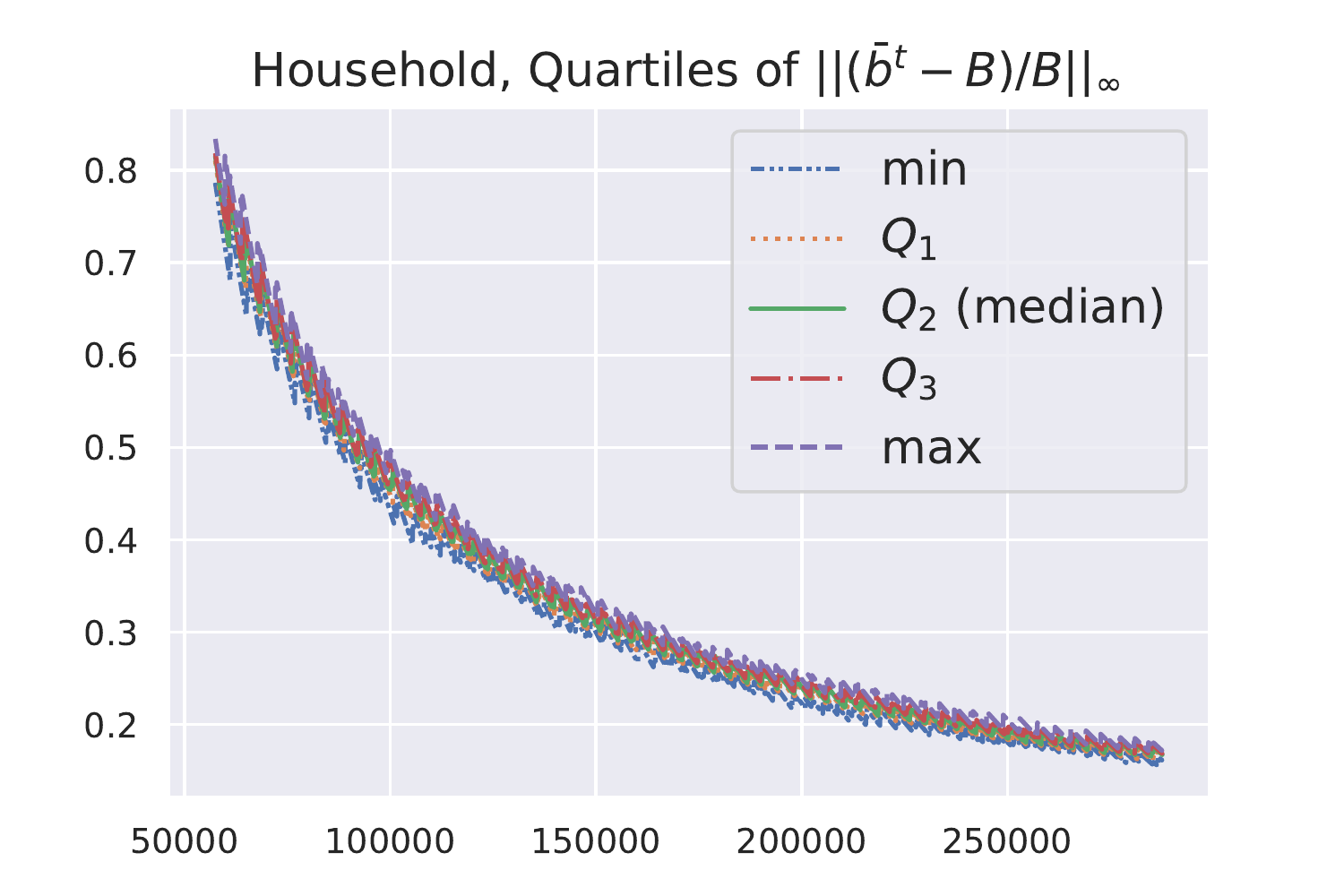} 

        \includegraphics[scale=0.35]{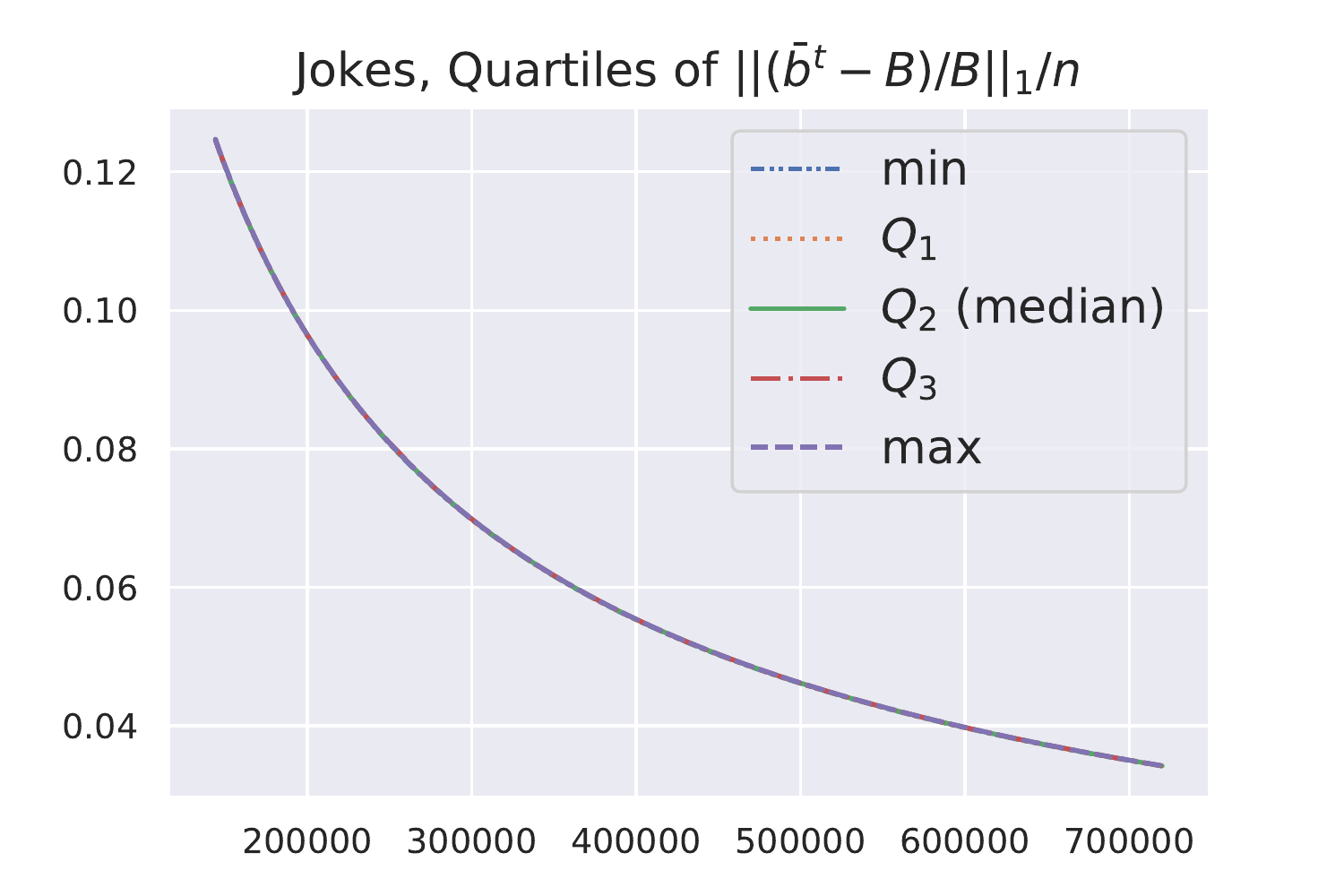}
        \includegraphics[scale=0.35]{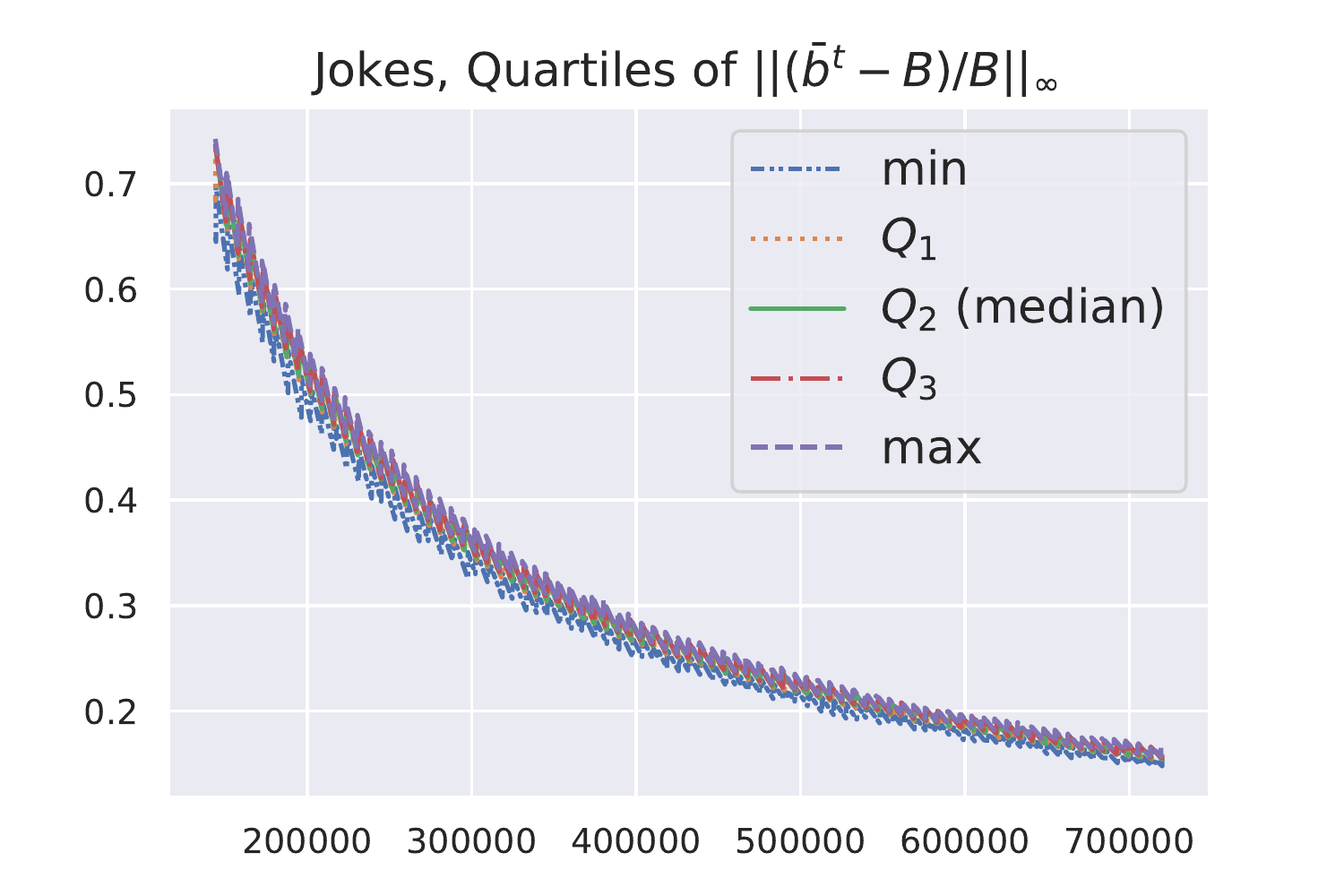}
        
        \includegraphics[scale=0.35]{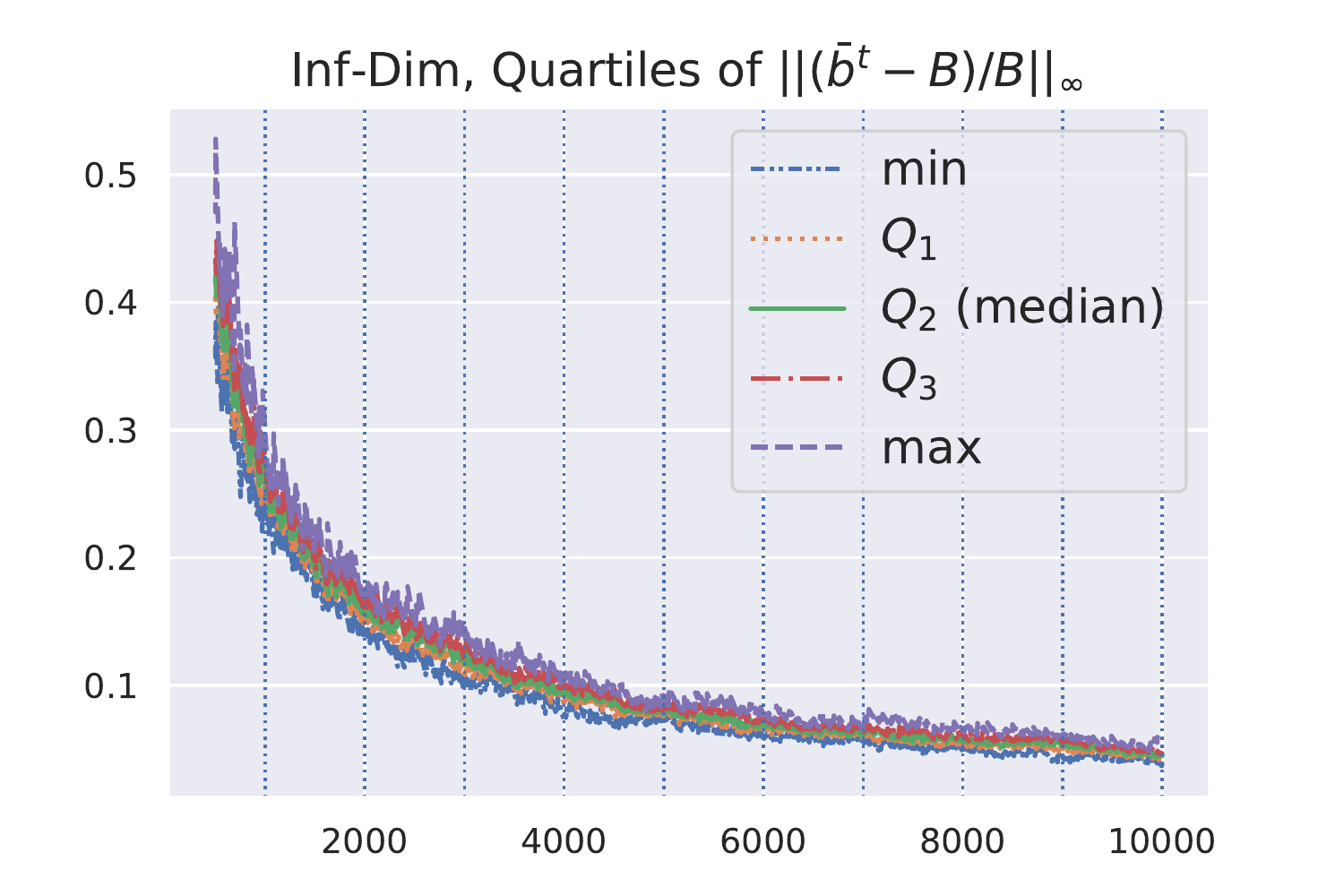}
        \includegraphics[scale=0.35]{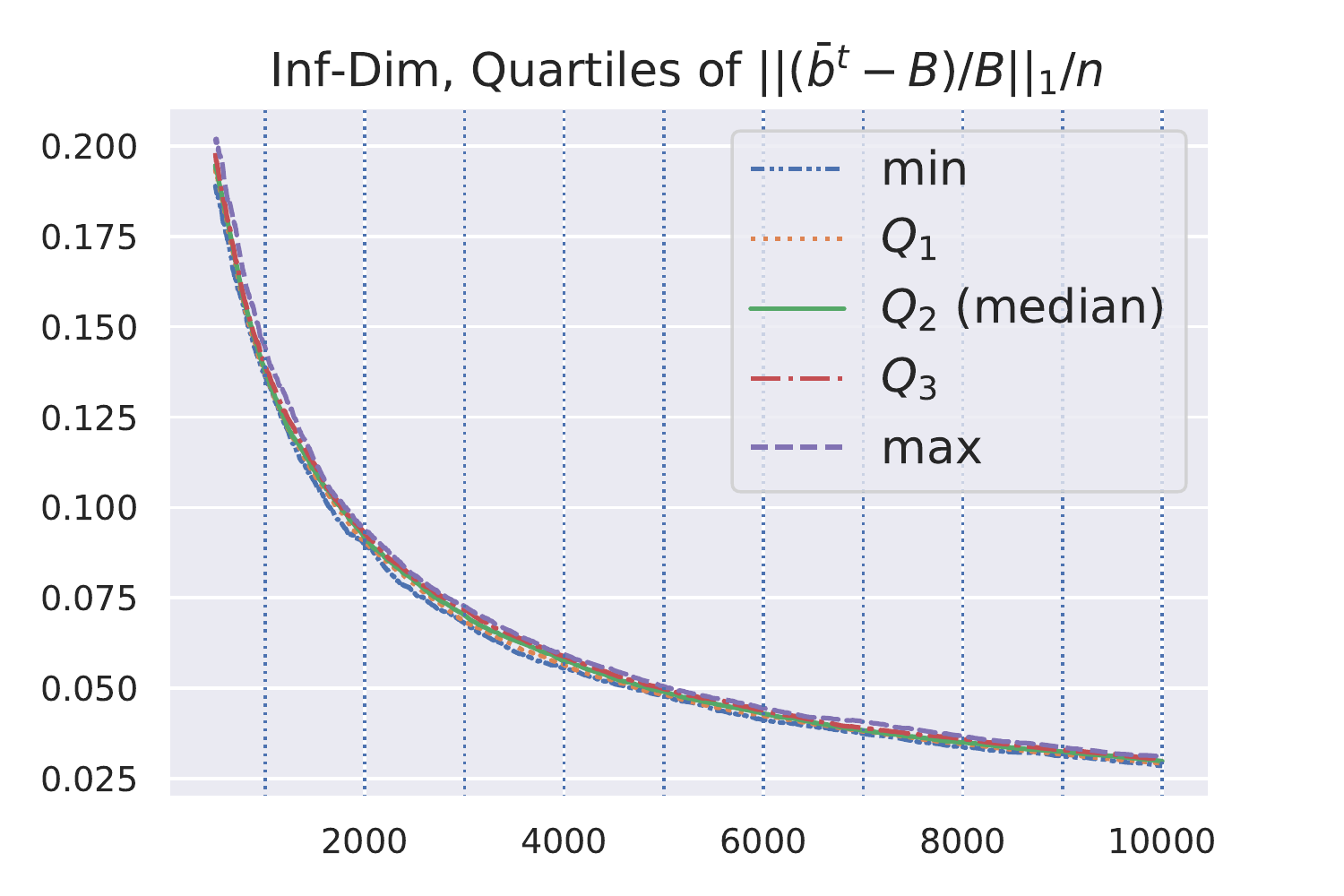}
        \caption{The PACE cumulative expenditure $\sum_{\tau=1}^t b^t_i$ of each buyer are close to the total amount of budget $tB_i$, as the quartile plots show. Vertical lines indicate when $t$ is a multiple of $10 n$. }
        \label{fig:quartiles-spending-vs-budget}
    \end{figure}
\end{center}

\end{document}